\documentclass[12pt]{article}	
\usepackage{amssymb,latexsym,rotate,epsf,graphicx, url, color, amsmath, amsthm,subcaption}
\usepackage[round, sort]{natbib}
\usepackage{subcaption,epstopdf}
\newtheorem{thm}{Theorem}

\newtheorem{lm}{Lemma}

\newtheorem{rem}{Remark}
\newtheorem{Ex}{Example}

\def\p{^{\prime}}

\def\1{^{[1]}}

\def\R{\mathbb{R}}

\DeclareMathSymbol{\I}{\mathord}{operators}{`i}
\DeclareMathSymbol{\e}{\mathord}{operators}{`e}

\newcommand{\dd}{\,\mathrm{d}}

\providecommand{\keywords}[1]{\textbf{#1}}

\begin{document}
\title{Negative interest rates: why and how?}
\author{Jozef Kise{\v l}\'ak, Philipp Hermann and  Milan Stehl\'\i k
}
\date{}
\maketitle

 \begin{abstract}
The interest rates (or nominal yields) can be negative, this is an unavoidable fact which has already been visible during the Great Depression (1929-39). Nowadays we can find negative rates easily by e.g.~auditing. Several theoretical and practical  ideas how to model and eventually  overcome empirical negative rates  can be suggested, however, they are far beyond a simple practical realization.
In this paper we discuss the dynamical reasons why negative interest rates can happen in the second order differential dynamics  and how they can influence the variance and expectation of the interest rate process.
Such issues are highly practical, involving e.g.~banking sector  and pension securities.

 \end{abstract}

\textbf{Keywords: }\keywords{negative interest rate, 2nd order dynamics, Wiener Process, Expectation,Variance}

\textbf{ AMS Subject Classification} 34F05, 91B70
\section{Introduction}
 \cite{Gesell} was the first to introduce the idea of negative interest rates equaling a ``carry tax'' on currency, such that holding money will affect the holder in a way to pay for this financial strategy. Principally, this tax would counteract to hoarding currency, because these hoarding costs would equal or exceed those costs resulting from negative interest rates such that one would rather deposit or lend currency than hoard it. One has to remind that this tax would only then be imposed if the zero bound became a constraint on monetary policy. Consequently, this monetary tool only has to be introduced rarely and temporarily limited, see \cite{Goodfriend}. Hereby, one has to consider two types of interest rates, i.e. nominal interest rates (yields), in a simple way explained as the payable rate, defined in contracts e.g.~corresponding to market prices, as well as real interest rates (also defined as the purchasing power of the interest payments received), where one has to subtract the inflation rate from nominal interest rates. Naturally, rates in a general consideration can be negative (see e.g.~\cite{Odekon, AndersonLiu} among others). In order to verify the payment of negative interest rates, \cite{Gesell} proposes physically stamping currency in the mid of 20$^{th}$ century, whereas \cite{Goodfriend} exploits the technological development indicating that this issue can be clarified with the aid of e.g.~magnetic strips on bills.
Recent discussion on investors safety with respect to negative interest rates (is given in e.g.~\cite{AndersonLiu}).

However, the fact that nominal interest rates become zero or deceed this threshold result from specific circumstances, where \cite{AndersonLiu} consider fear or uncertainty as main forces. \cite{Buiter} introduce the events of 11 September 2001 as additional possible external cause for a constraint of the nominal interest rate to the lower bound. Therefore, \cite{Ahearne} have studied Japans experience with this situation in order to ``discuss the zero bound problem'' and to prevent deflation. However, negative interests are not seen positively from any perspective, because e.g.~\cite{AndersonLiu} criticize negative interest rates with the aid of two arguments: i) the rates would be set to a negative value only when economic conditions are so weak that the central bank has previously reduced its polity rate to zero and ii) negative interest rates may be interpreted as a tax on banks - a tax that is highest during periods of quantitative easing. Generally, \cite{Buiter} consider that there exist two possibilities to either avoid or escape from zero bound trap. On the one hand they suggest to wait and hope for some positive shock to the effective demand of goods and services and on the other hand they propose to tax currency. Moreover, the authors define that the rate of returns on money, i.e. for coin and currency, can be zero, however, for liabilities of private deposit-taking institutions, which are most of the broader monetary aggregates, this rate will generally be positive.

The literature, see e.g.~cite{Goodfriend, Cecchetti, Buiter} among others, has been focusing on the issue of negative interest rates, which happened empirically in the 1930s and early 1940s for U.S. Treasury bonds or by two European central banks as Riksbank (Sweden) and Danmark Nationalbank (Denmark), respectively. The two Scandinavian banks operate with three policy rates, i.e. term deposit rate, overnight repo rate and lending rate. Occasionally the less important term deposit rate was set below zero by these banks. Media have shed light on this special case, however, due to the very small amounts in negative term deposits rates, in this case for a 7-day and 14-day period deposit respectively, this issue was not that important from the viewpoint of banks (see \cite{AndersonLiu}). These and similar occurrences lead to scholarly interest in overcoming the zero bound such that monetary policy makers fear that standard monetary tools do not operate properly under these circumstances, see \cite{Buiter}. These circumstances raised concern of e.g.~\cite{Keynes} who started to work on the consequences for macroeconomics and monetary policy of the zero bound on nominal interest rates. These thoughts have been reassumed by \cite{Summers91} or \cite{Fisher}, who argue that an inflation exceeding 3\% should be targeted in order to gain scope for a fall of three percentage points of nominal interest rate before hitting the zero bound. \cite{Goodfriend} refers to several studies which have shown that significantly higher costs will arise from higher inflation. Therefore, one might have to consider other solutions, being more satisfactory,  to this problem.

Generally,  \cite{Goodfriend} claims that an introduction of monetary tools such as a ``carry tax'' is only then performable, when programs focusing on introducing this methodology to the public accompany. Therefore, portfolios of individuals should not be heavily dependent on short-term interest rates or generally programs shall clarify the reasonability of negative interest rates for the public. Economic education along the introduction of this tax is equally important as the development and installing the system.

Central banks operate in an economy where fluctuating interest rates keep the business cycle running properly. Hence, these can perform in a way to adjust real interest rate movements in order to target their desired short-term nominal interest rates. When considering zero expected inflation, a zero bound on nominal interest rates would expect real interest rates to be nonnegative. However, this concept of zero bounds creates considerable problems as i) negative interest rates have helped the economy, e.g.~in the Great Depression, in order to recover in a way that these periods are shorter and the effect will only develop to a smaller extent and ii) deflation can raise the expected real interest rates to adjust when nominal rates are at the threshold of zero, see \cite{Goodfriend}. When considering such situations one has to be aware of the fact that storage as well as carry costs accompany with paying negative interest rates on base money. Aforenamed costs are very small yielding that in \cite{Buiter} these are neglected, however, interest rates need to be adjusted by latter named costs which also include insurances against loss, damage or theft. In this context one has to consider that rational economic agents hold those store of values which will give them the best rate-of-return, which is base money as the most liquid store of value. Moreover, these agents will only then choose an alternative store of value which is at least as high as base money. Hence, nominal interest rate on bonds needs to be higher than the carry cost differential. In further consequence this will lead to an adjustment of interest rates for carry costs leading to relevant net financial rates of return.  \cite{Buiter} recall the fact that wealth is more frequently held in currency the poorer persons are, which is leading to a regressive effect for taxing currency. Moreover, they claim that heavily cash-based grey, black and outright criminal economies would then also be taxed when introducing this approach.

Different methods have been proposed by \cite{Cecchetti}  or \cite{Goodfriend} to approach negative interest rates. On the one hand three options have been introduced as: a carry tax on money, open market operations in long bonds, or monetary transfer (see \cite{Goodfriend}) and on the other hand Old-Keynesian or other New-Keynesian analytical models are described as a remedy for negative interest rates (see \cite{Buiter}). These old- and new-Keynesian models use a conventional and forward-looking IS curve, respectively, as well as backward-looking and forward-looking accelerationist Phillips curve, respectively. These models theoretically show that an increase of inflation rates, as proposed by \cite{Bryant} and \cite{Freedman}, will not help to overgo the problem of nominal interest rates at the zero bound. For a more detailed discussion on the application of real GDP, short real rate of interest, nominal stock of currency, general price level and nominal interest rate in the old- respectively new-keynesian model, see \cite{Buiter}. \cite{Goodfriend} claims that imposing a tax on holding money will enable to go beyond the threshold of zero interest rates, hence, one would rather accept negative nominal interest rates than paying for the allowance to hold money. Therefore, introducing this ``carry tax'' would be an option how to deal with negative interest rates. Due to the technological improvement an implementation of this concept is easily realizable. Moreover, the same methodology would also be applicable in terms of currency or exchange rates. The author considers a variable carry tax to be a powerful policy instrument, dependent on monetary aggregate targets of central banks in the case of very low short-term interest rates.

\cite{Buiter} coincide with these authors that this approach is feasible and they claim that it has potential to be more efficient than targeting inflation rates, i.e. that high inflation rates are remedy to minimize the risk of zero interest rates. Moreover, they point out that from the perspective of central banks, administrative problems arise  due to the anonymity when holding money. This fact complicates paying negative interest rates for holding money, where a proper method has to be found such that the currency owner has to reveal himself to pay the ``carry tax''. In contrast to that the ``carry tax'' is easily collected when focusing on electronic commercial bank reserves, where holders are known such that either positive or negative returns can be withdrawn from the deposits.  Hence, imposing this tax yields little or no costs for commercial banks' balances but some administrative costs, possibly to a greater extent, for currency  can result from it. Nevertheless, there are two reasons why interest is not paid on (coin and) currency: i) attractions of seigniorage and ii) administrative difficulties of paying a negative interest rate on bearer bonds, which are debt securities in paper or electronic form whose ownership is transferred by delivery rather than by written notice and amendment to the register of ownership. These are negotiable to the same extent as other money market instruments such as bank certificates of deposit or bills of exchange. Coin and currency can be seen as bearer bonds issued by central banks in our setup. For both, the state and private agents, significant costs arise due to fundamental information asymmetry (see \cite{Buiter}). For reasons of clarification it is mentioned that paying negative interest rates, as introduced in \cite{Buiter}, is exactly the same as the ``carry tax'' proposed by \cite{Goodfriend} to overcome the zero bound.

\cite{Krugman} actually suggests that central banks, in the case when the zero bound is hit, should operate in a way to adjust the rate of inflation for some period of time such that the real interest rate becomes negative. However, the author does not provide a more detailed explanation to approach this target. He was criticized in \cite{Goodfriend} for not considering that central banks with this power could directly stimulate spending instead. From the perspective of central banks acting in a way only to react on the basis of inflation rates is  practically a very challenging one, because on the one hand knowing when and on the other hand to which extent adjustments are needed would be of major importance.  However, \cite{McCallum} proposes to use the foreign exchange rate to be a powerful policy tool when central banks have to find remedy when targeting the zero threshold. More precisely, with this policy rule the author claims that one can stabilize inflation and output. For a more detailed discussion on the effects of exchange rate depreciation leading in higher net exports such that nominal short interest rates adjust immediately foreign currency interest rates, see \cite{McCallum}.

Removing storage costs or user fees for electronic reserve balances collected from central banks will naturally set a zero interest floor on interbank interest  rate without respect of physical costs of storing (any) currency. Hence, short rates are pulled down to zero due to the competition between banks and in further consequence long-term rates will move down with the average of expected future short rates over the relevant horizon, see \cite{Goodfriend}.

\cite{Buiter} discuss the liquidity trap, describing situations where private agents absorb any amount of real money ceteris paribus, i.e. not changing their behavior to any extent. However, they recall modern theories where the short riskless nominal interest rate on government debt can be interpreted as opportunity costs of holding currency. In this case the lower threshold of nominal interest rates is seen to be zero. These two monetary policy tools (liquidity traps and lower bound on nominal interest rates) were of great interest during periods of high inflation such as 1970s or 1980s. In \cite{McCallum} liquidity traps are described as situations in which central bank's usual policy instrument cannot be lowered past a prevailing zero lower bound (or possibly some negative lower bound). This leads to the fact that potential stabilizing powers of monetary policy can be nullified by the occurrence of a ``liquidity trap''.

\cite{Evans} have proposed that rational expectations (RE) need to be assumed when modeling expectations. Moreover, the authors claim that rational expectations are modeled with the aid of conditional (on available information for the decision makers) expectations on the basis of relevant variables. In further consequence when modeling rational expectations one also has to consider ``adaptive learning'', which represents the adjustment of the forecast rule of agents when new data becomes available over time. The authors provide an example where agents have to periodically relearn the relevant stochastic processes when an economy occasionally ``undergoes structural shifts''. These shifts accompany with ``overcoming the zero bound'', when e.g.~applied on the discussed case.

\cite{Buiter} stress that rational expectations equilibria are only then of economical interest, when they can be seen as E-stable, which means that they are asymptotically stable under least squares learning. Following to that they address that learning rules as given by \cite{Evans} need to be evidently applicable due to e.g.~cognitive psychology, showing that agents actually tend to behave as described by the model, which still has to be shown.

\cite{Stehlik15} introduced a modified Parker second order model of interest rates. Thereby, their first example refers to ``Oscillatory interest rate-deterministic part''. This method enables to compute negative interest rates with the aid of the Parker Equation as a solution for this setup. Another broadly used model which can operate with negative interest rates is introduced by \cite{Vasicek}. More precisely this Ornstein-Uhlenbeck (Gaussian) process is often used to derive equilibria models for discount bond prices.
However, the literature, e.g.~\cite{Chan}, proposes to interpret the interest rate as a real but not as a nominal one in this model and frequently criticizes these approaches for the possibility of negative rates as well as the implication of homoscedastic interest rate changes.
In the next section we introduce the realistic and relatively simple model of interest rate which has a potential to understand dynamics of negative rates.

\section{Modelling of interest rate}

Assume that the model describing the future evolution of interest rates
$r_t$ can be written, for $t\in I=[0,\infty)$, by
\begin{eqnarray}\label{eq:stochsys}
 {\mathrm d}r_t&=&c(t)\,(p_t)^m\,{\mathrm d}t \nonumber\\
 {\mathrm d}p_t&=&\left[a(t)\,(p_t)^l+b(t)\,(r_t)^n\right]{\mathrm d}t+\sigma(t)\,(r_t)^k\,{\mathrm d}W_t,
\end{eqnarray}
where $a,b,c,\sigma\in C(I)$ and $n,k,l\in \mathbb{Q}, ~m\in\mathbb{Q}\setminus \{0\}$ such that their denominators are odd numbers. This
 simplification is used to avoid problems with the definition of the domain of a power function.
Otherwise a so-called signed power function $\Phi(z):=|z|^{\alpha-1}\,z, ~\alpha, ~z\in\mathbb{R}$ can be used
(if necessary it is used throughout the article without changing the denotation).
 Moreover, we assume that $\sigma$ is positive and $c(t)\not\equiv 0, ~c(0)\ne 0$.
This model naturally generalizes the linear model in \cite{Parker} $~(c(t)\equiv 1, a(t)\equiv a, b(t)\equiv b, \sigma(t)=\sigma, m=n=l=1, k=0)$ and also
two nonlinear models studied in \cite{Stehlik15}.
Using the classical result, see Theorem \ref{exi} in the Appendix, we can directly obtain unique results.
In our case we have $n=2, ~~X_t=(r_t,s_t)$ and 	
$$b(t,X_t)=\left(c(t)\,(p_t)^m, ~a(t)\,(p_t)^l+b(t)\,(r_t)^n\right)^T,$$
$$\Sigma(t,X_t)=\left(\begin{array}{cc}
0 & 0\\
0 &\sigma(t)\,(r_t)^k
\end{array}\right).
$$
Existence is always secured for the linear case, but not in general.
Changing values of $k, l, m$ or $n$ either violates uniqueness or causes a blow-up.
The fundamental tool for transformations of SDE's is It\^o's Lemma \ref{ito} (the version
given in Appendix is due to \cite{Oksendal}.
With Lemma \ref{ito} we are ready for the setup of a transformation to remove
the level dependent noise. It{\^o}'s formula can be used to solve SDE's,
although the class of equations that are solvable in this fashion is limited. Next Theorem \ref{t:tran} helps us to find explicit solutions
of the problem \eqref{eq:stochsys} in specific cases and could be a useful supplementary tool for solving SDE's.

\begin{thm}\label{t:tran}
 Let $(r_t,p_t)$ be a diffusion process as in \eqref{eq:stochsys}, then the
 transformation $$\psi_1(s_t,r_t,t)=r_t, ~~z_t=\psi_2(p_t,r_t,t)=\frac{t}{2}-W_t+\frac{p_t}{\sigma(t)\,(r_t)^k}$$
 will result in a deterministic differential system involving Wiener process (ODE with random coefficients)
\begin{eqnarray}\label{eq:detsys}
 {\mathrm d}r_t&=&c(t)\,\sigma(t)^m\,(r_t)^{km}\left(u_t\right)^m\,{\mathrm d}t \\
 {\mathrm d}z_t&=&\left[\frac{1}{2}-\frac{\sigma\p(t)}{\sigma(t)}\,u_t-k\,c(t)\,\sigma(t)^m\,(u_t)^{m+1}(r_t)^{mk-1}+a(t)\,\sigma(t)^{l-1}\,(u_t)^l\,(r_t)^{k(l-1)}+\frac{b(t)}{\sigma(t)}\,(r_t)^{n-k}\right]{\mathrm d}t,\nonumber
\end{eqnarray}
where $u_t=z_t+W_t-\frac{t}{2}.$
\end{thm}

\begin{rm}
Since our model does not only consist of SDE but also of an initial conditions,
we have to specify them. Usually $r_0=A, ~0<A\ll 1$ and $\left.\frac{\mathrm\mathrm{d}r_t}{\mathrm\mathrm{d}t}\right|_{t=0}=c(0)\,(p_0)^m=:B,
~~ 0\leq |B|\ll1$.
Then after the transformation used in Theorem \ref{t:tran} we have
$z_0=\frac{B^\frac{1}{m}}{c(0)^\frac{1}{m}\sigma(0)A^k}:=\tilde B$.
\end{rm}

\begin{Ex}[Linear case]
For $k=0, ~l=m=n=1, ~a(t)\equiv a, ~b(t)\equiv b, ~c(t)\equiv c$
the solution of the SDE can be studied in terms of the roots of the
characteristic equation of the process, similarly as in \cite{Parker}.
Three possibilities can occur due to quadratic equation with discriminant $D=4\,c\,b+a^2$.
\begin{itemize}
 \item $D=0$ - real and equal,
 \item $D>0$ - real and distinct,
 \item $D<0$ - complex conjugate.
\end{itemize}
The expected values and autocovariance functions of the force can be derived directly
as explicit solution can be found or by using the linearity of the process.
\end{Ex}

\subsection{Case study I}
  Case: $k=0, l=1, b(t)\equiv 0$. Then equation \eqref{eq:stochsys} reduces to a
  simple nonlinear model
 \begin{eqnarray}\label{eq:kab0}
 {\mathrm d}r_t&=&c(t)\,(p_t)^m\,{\mathrm d}t \nonumber\\
 {\mathrm d}p_t&=&a(t)\,p_t\,{\mathrm d}t+\sigma(t)\,(r_t)\,{\mathrm d}W_t,
\end{eqnarray}
with the corresponding deterministic system
\begin{eqnarray}\label{eq:syskab0}
 {\mathrm d}r_t&=&c(t)\,\sigma(t)^m\left(u_t\right)^m\,{\mathrm d}t \nonumber\\
 {\mathrm d}z_t&=&\left[\frac{1}{2}-\frac{\sigma\p(t)}{\sigma(t)}\,u_t+a(t)\,u_t\right]{\mathrm d}t.
\end{eqnarray}
Several Figures 1-5 illustrate the dynamics of the process determined by the system \eqref{eq:kab0}.
Here notice that a trajectory of $r_t$ is smooth but its derivative possesses  non-smoothness of the trajectory.
A general solution can be found directly (using It\^o calculus for \eqref{eq:kab0}),
or using Theorem \ref{t:tran}, since the second equation in \eqref{eq:syskab0} is linear 1st order ODE
not depending on $r_t$ explicitly and the first one can be subsequently integrated. Using initial conditions we get the following result.

\begin{lm}\label{lm:e}
Denote a process $R_t$ as the solution of Cauchy problem \eqref{eq:stochsys} with
$k=0, l=1, ~b(t)\equiv 0, ~m>0$ and $R_0=A, ~R\p_0=B$, then
$$R_t=A+\int _{0}^{t}\! \left[ {{\rm e}^{-\int _{0}^{u}{\frac {{\frac {d}{
ds}}\sigma \left( s \right) }{\sigma \left( s \right) }}-a \left( s
 \right) {ds}}} \left( \frac{1}{2}\,\int _{0}^{u}\!{{\rm e}^{\int _{0}^{s}
{\frac {{\frac {d}{dv}}\sigma \left( v \right) }{\sigma \left( v
 \right) }}-a \left( v \right) {dv}}} \left( 1+ \left( -{\frac {{
\frac {d}{ds}}\sigma \left( s \right) }{\sigma \left( s \right) }}+a
 \left( s \right)  \right)  \left( 2\,W_s -s \right)
 \right) {ds}+ \right.\right.$$
 $$+\left.\left.\left( {\frac {B}{c \left( 0 \right) }} \right) ^{\frac{1}{m}}  \sigma \left( 0  \right) ^{-1} \right) +W_u -\frac{u}{2}\right] ^{m} \sigma \left( u \right)
^{m}c \left( u \right) {du}
.$$ Moreover,
\begin{itemize}
 \item For $B=0, \sigma(t)\equiv \sigma, a(t)\equiv 0$
$$E[R_t]=\begin{cases}
  \displaystyle A+\sigma^m\,L(m;c),& \mbox{if m is even},\\
 A, & \mbox{if m is odd}
 \end{cases}$$
 and
 $$Var[R_t]=\sigma^{2m}\left[\int_0^t\int_0^tc(s)\,c(u)\,\mathcal{E}_m(s,u)\dd s\dd u-L^2(m;c)\right],$$
 where
 $$L(m;c):=\frac{m!}{2^\frac{m}{2}(\frac{m}{2})!}\int_0^t c(u)\,u^\frac{m}{2}\mathrm\mathrm{d}u,$$ and
 $$ \mathcal{E}_m(s,u)=\begin{cases}
   \displaystyle\frac{(m!)^2}{2^{m}}\sum_{j=0}^\frac{m}{2}\frac{(2\frac{\min{(s,u)}}{\sqrt{s\,u}})^{2j}}{(2j)!\left[\left(\frac{m}{2}-j\right)!\right]^2}, & \mbox{if} ~~m ~\mbox{is even},\\
   \displaystyle\frac{(m!)^2}{2^{m}}\sum_{j=0}^\frac{m-1}{2}\frac{(2\frac{\min{(s,u)}}{\sqrt{s\,u}})^{2j+1}}{(2j+1)!\left[\left(\frac{m-1}{2}-j\right)!\right]^2}, & \mbox{if} ~~m ~\mbox{is odd}.\\
  \end{cases}$$
 \item For $B\ne 0, \sigma(t)\equiv \sigma, a(t)\equiv 0$
$$E[R_t]=A + \sigma^m \sum_{j=0}^{\lfloor \frac{m}{2}\rfloor } \binom m{m-2j} \left(\frac{c_1}{\sigma}\right)^{m-2j} \frac{(2j)!}{2^j j!} \int_0^t c(u)u^j \dd u,$$
 $$Var[R_t]\approx A^2+2Am\,c_1^{m-1}\sigma\int_0^t c(u)\,\mathrm{d}u+
 m^2\,\sigma^2\,c_1^{2m-2}\left(\int_0^t c(u)\,\mathrm{d}u\right)^2+$$
 $$+ c_1^{2m}\int_0^t\int_0^t c(u)\,c(s)\,\min{(u,s)}\,\mathrm{d}u\mathrm{d}s,$$
 where $c_1= \left({\frac {B}{c \left( 0 \right) }} \right) ^{\frac{1}{m}}$.
 \end{itemize}
 \end{lm}

\begin{proof}
The form of the process $R_t$ can be found directly or by using Theorem \ref{t:tran}.
We have to split the proof in two cases in dependence of parity of $m$.
For  $B=0$ we have
$$
R_t = A + \sigma^m \int_0^t c(u)\,W_u^m \dd u .
$$
Since we can interchange the order of expectation and integration, we have
\begin{align*} E[R_t] &= A + \sigma^m \int_0^t E[W_u^m]c(u)\dd u \underset{\text{Lemma \ref{lemmastar}}}{=}
 A + \sigma^m \int_0^t \frac{m!}{2^{\frac m2} (\frac m2)!} u^{\frac m2} c(u) \dd u \\
&= A + \sigma^m L(m;c), \quad \text{if m is even}.
\end{align*}
Using Lemma \ref{lemmastar} we directly obtain that $E[R_t] = 0$ holds for odd $m$.

Further, we are able to compute variances. One can use Isserlis theorem, but computationally it is more conveniente to use
Lemma \ref{fan}, see \cite{Isserlis18} or \cite{Ken77}[p. 94].
\begin{align*}
E[R_t^2] &= A^2 + 2A\sigma^m L(m;c) + \sigma^{2m}E\left[\left(\int_0^t c(u) W_u^m \dd u\right)^2\right] = \\
&A^2 + 2A\sigma^m L(m;c) +  \sigma^{2m}\int_0^t\int_0^tc(s)\,c(u)\,E[W^m_s\,W^m_u]\dd s\dd u \underset{\text{Lemma \ref{fan}}}{=} \\
&A^2 + 2A\sigma^m L(m;c) +\sigma^{2m}\int_0^t\int_0^tc(s)\,c(u)\,\mathcal{E}_m(s,u)\dd s\dd u,
\end{align*}
\begin{align*}\Rightarrow Var(R_t) &= E[R_t^2] - A^2 - 2A\sigma^m L(m;c) - \sigma^{2m}L^2(m;c) = \\
&=  \sigma^{2m}\left[\int_0^t\int_0^tc(s)\,c(u)\,\mathcal{E}_m(s,u)\dd s\dd u-L^2(m;c)\right].\end{align*}

For $B\ne 0$
\begin{align*}
R_t &= A + \int_0^t c(u) [W_u\sigma + c_1]^m \dd u, \quad c_1 =
\frac{B^{\frac 1m}}{c(0)^{\frac 1m}}\\
E[R_t] &= A + \int_0^tc(u)E\left[(W_u \sigma + c_1)^m\right]\dd u =\\
&= A+\int_0^t c(u)E\left[\sum_{k=0}^m \binom mk c_1^k\,{(W_u\sigma)^{m-k} \dd u}\right]= \\
&= A + \int_0^t c(u) \sigma^m \sum_{k=0}^m \binom mk \left(\frac{c_1}{\sigma}\right)^k E[W_u^{m-k}]\dd u = \\
&= A + \sigma^m \sum_{k=0}^m \binom mk \left(\frac{c_1}{\sigma}\right)^k \int_0^t c(u) E[W_u^{m-k}] = \\
&= A + \sigma^m \sum_{l=0}^m \binom m{m-l} \left(\frac{c_1}{\sigma}\right)^{m-l} \int_0^tc(u)E[W_u^l]\dd u =\\
&= A + \sigma^m \sum_{j=0}^{\lfloor \frac{m}{2}\rfloor } \binom m{m-2j} \left(\frac{c_1}{\sigma}\right)^{m-2j} \int_0^t c(u) u^j \frac{(2j)!}{2^j j!} \dd u =\\
&= A + \sigma^m \sum_{j=0}^{\lfloor \frac{m}{2}\rfloor } \binom m{m-2j} \left(\frac{c_1}{\sigma}\right)^{m-2j} \frac{(2j)!}{2^j j!} \int_0^t c(u)u^j \dd u
\end{align*}

Since we cannot use Lemma \ref{fan} for $B\ne 0$ , we use Taylor approximation of integrand of $R_t$ instead. We have
$$R_t\approx A+\int_0^tc(u)\left[m\,c_1^{m-1}\sigma+c_1^m\,W_u\right]\,\mathrm{d}u\Rightarrow$$
$$R_t^2\approx A^2+2A\int_0^tc(u)\left[m\,c_1^{m-1}\sigma+c_1^m\,W_u\right]\,\mathrm{d}u+
\left(\int_0^tc(u)\left[m\,c_1^{m-1}\sigma+c_1^m\,W_u\right]\,\mathrm{d}u\right)^2.$$
Now,
$$E[R_t^2]\approx A^2+2A\,m\,c_1^{m-1}\sigma\int_0^tc(u)\,\mathrm{d}u+m^2\,\sigma^2\,c_1^{2m-2}
\left(\int_0^tc(u)\,\mathrm{d}u\right)^2+c_1^{2m}E\left[\left(\int_0^tc(u)\,W_u\,\mathrm{d}u\right)^2\right]$$
gives us the result.

\end{proof}

\begin{rem}
 Similar result can be analogously derived for more general but fixed $m$.
\end{rem}

\begin{figure}[h!]\label{fig:cos}
\centering
\begin{subfigure}[b]{.45\linewidth}
   \includegraphics[width=0.99\textwidth]{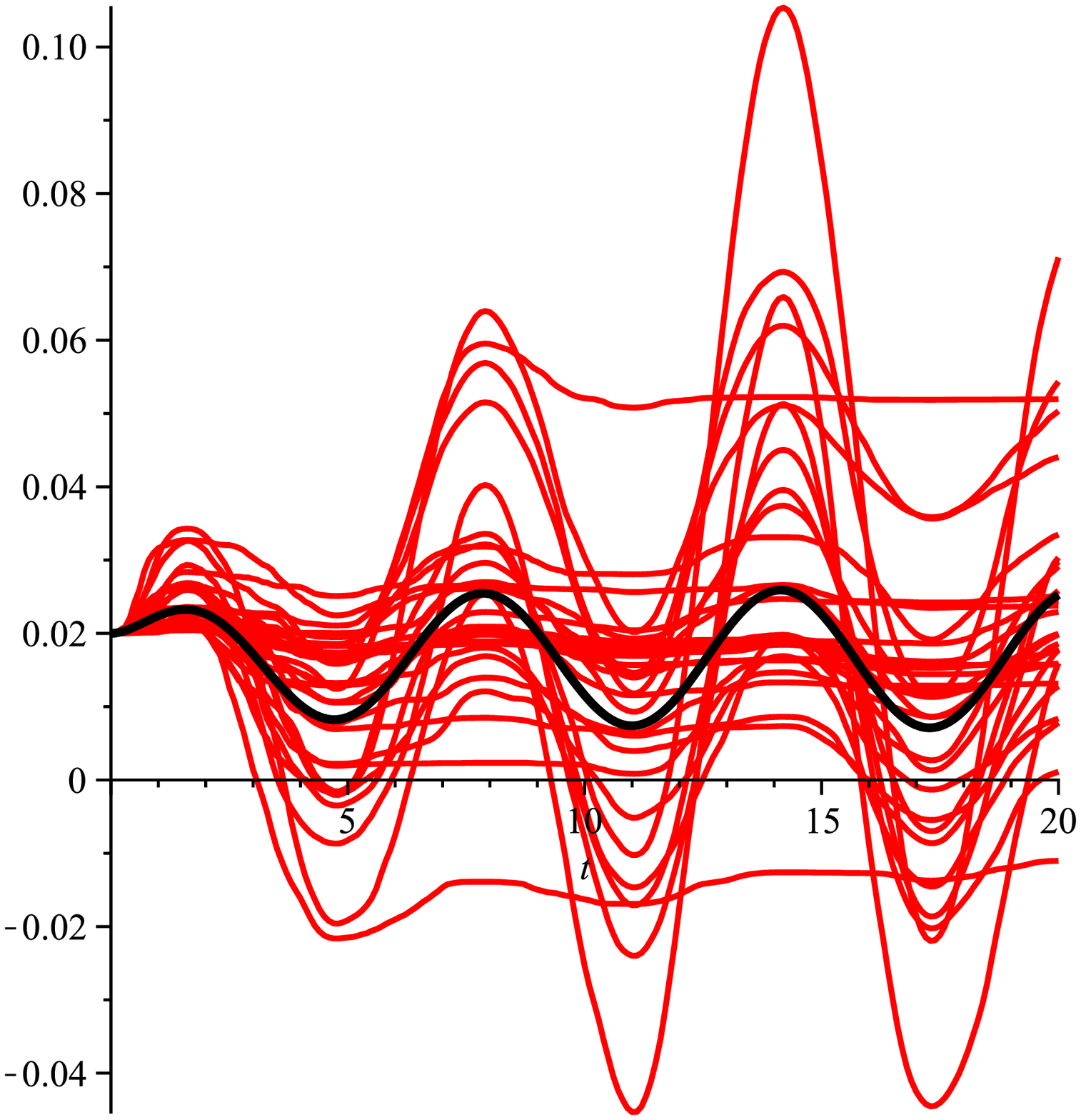}
\label{ob:cos}
\caption{$c(t)=\frac{\cos{t}}{1+t}$.}
 \end{subfigure}
\begin{subfigure}[b]{.45\linewidth}
   \includegraphics[width=0.99\textwidth]{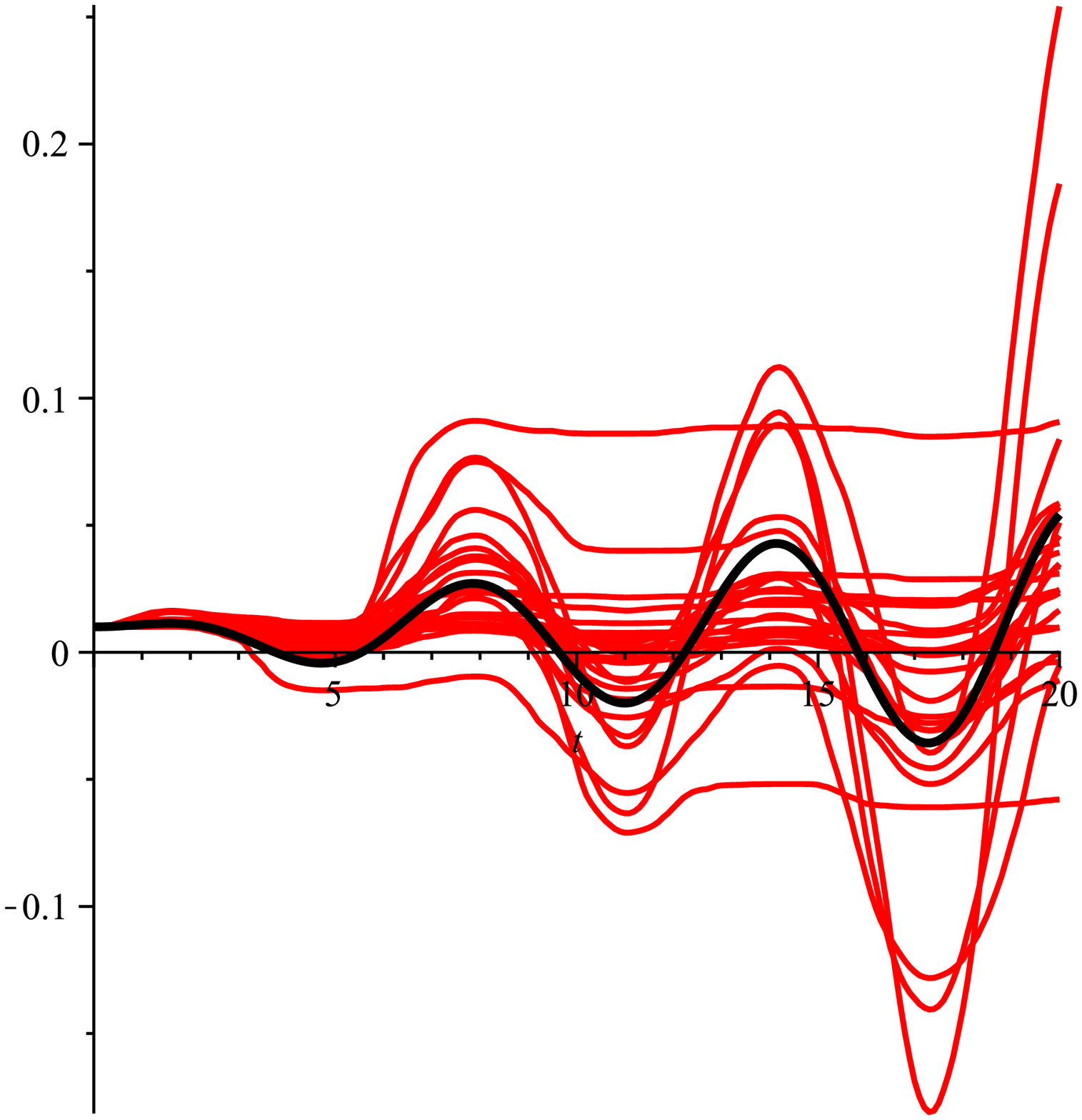}
\label{ob:cos2}
\caption{$c(t)=\cos{t}$.}
 \end{subfigure}
 \caption{Expected value (black) and 25 realization (red) of process $R_t$ defined by \eqref{eq:kab0} with $k=0, ~A=0.02, ~B=0, ~m=2, ~a(t)=b(t)\equiv 0, ~\sigma(t)\equiv 0.05$.}
\end{figure}

\begin{figure}[h!]
\centering
\begin{subfigure}[b]{.45\linewidth}
   \includegraphics[width=0.99\textwidth]{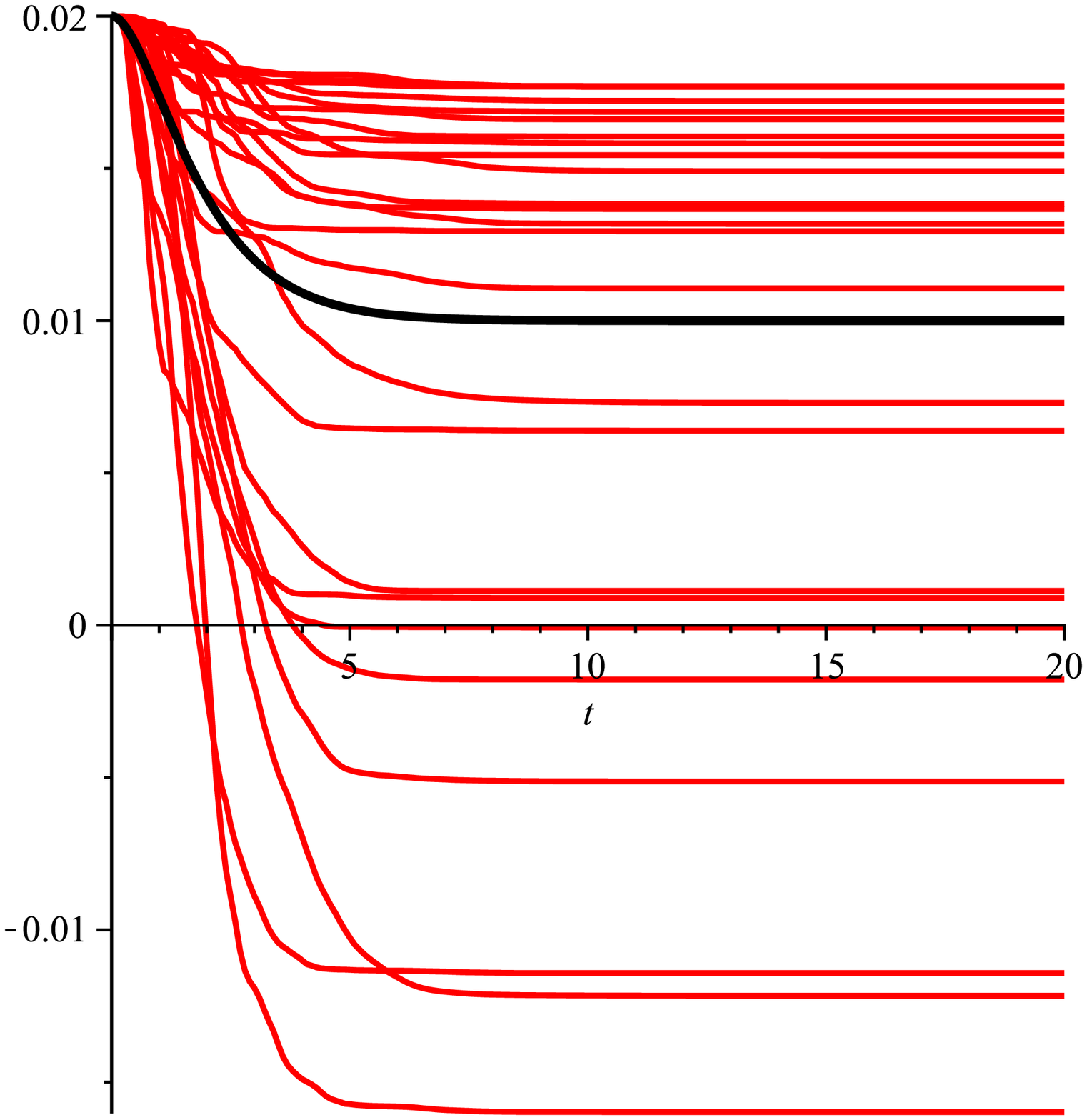}
\subcaption{$\sigma=0.05$. \label{fig:s}}
 \end{subfigure}
\begin{subfigure}[b]{.45\linewidth}
   \includegraphics[width=0.99\textwidth]{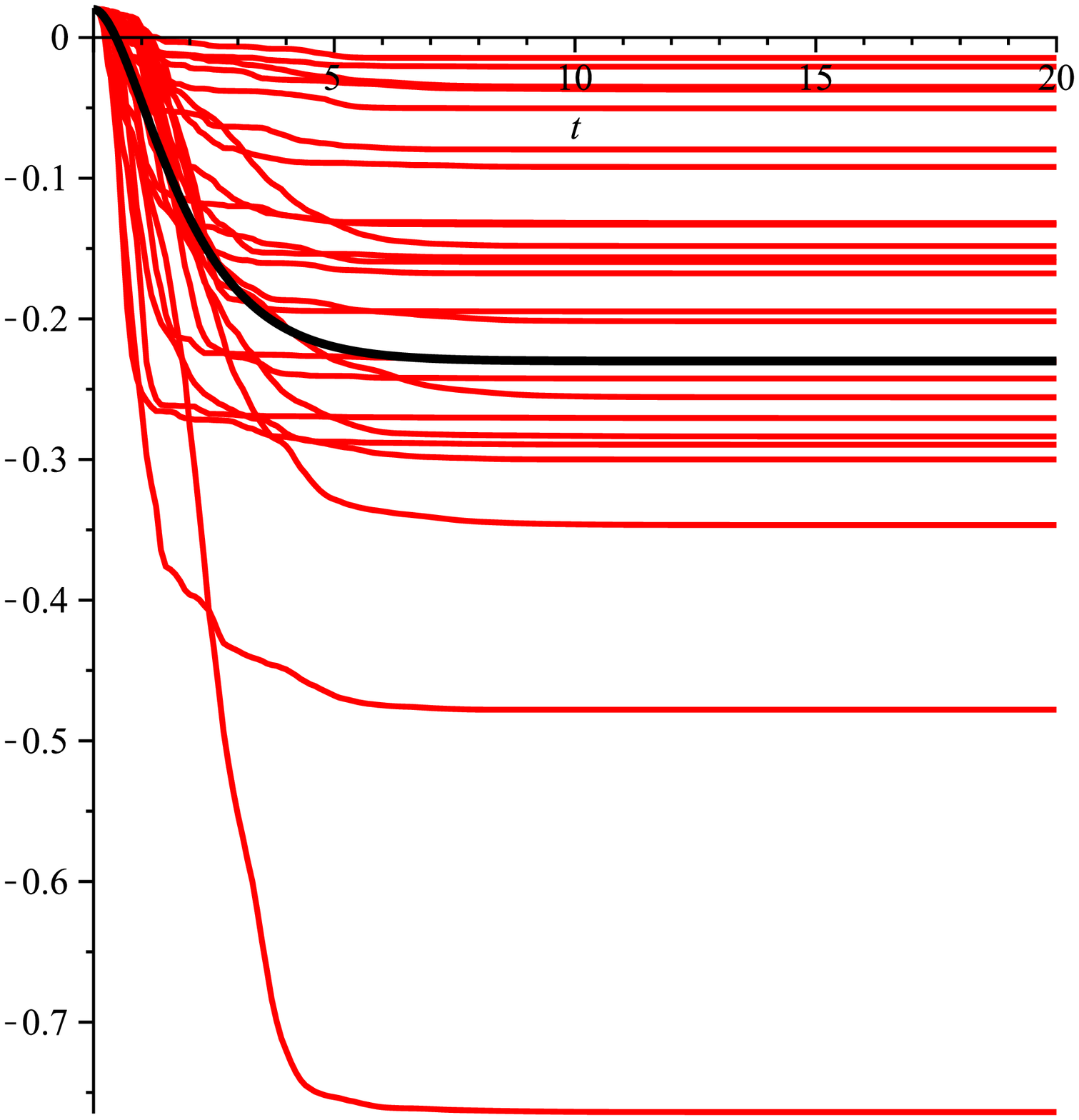}
\subcaption{$\sigma=0.1$.}
 \end{subfigure}
 \caption{Expected value (black) and 25 realization (red) of process $R_t$ defined by \eqref{eq:kab0} with $c(t)=-\mathrm{e}^{-t}, ~k=0, ~A=0.02, ~B=0, ~m=2, ~a(t)=b(t)\equiv 0$.}
\end{figure}

\begin{figure}[h!]
\centering
\begin{subfigure}[b]{.3\linewidth}
   \includegraphics[width=0.99\textwidth]{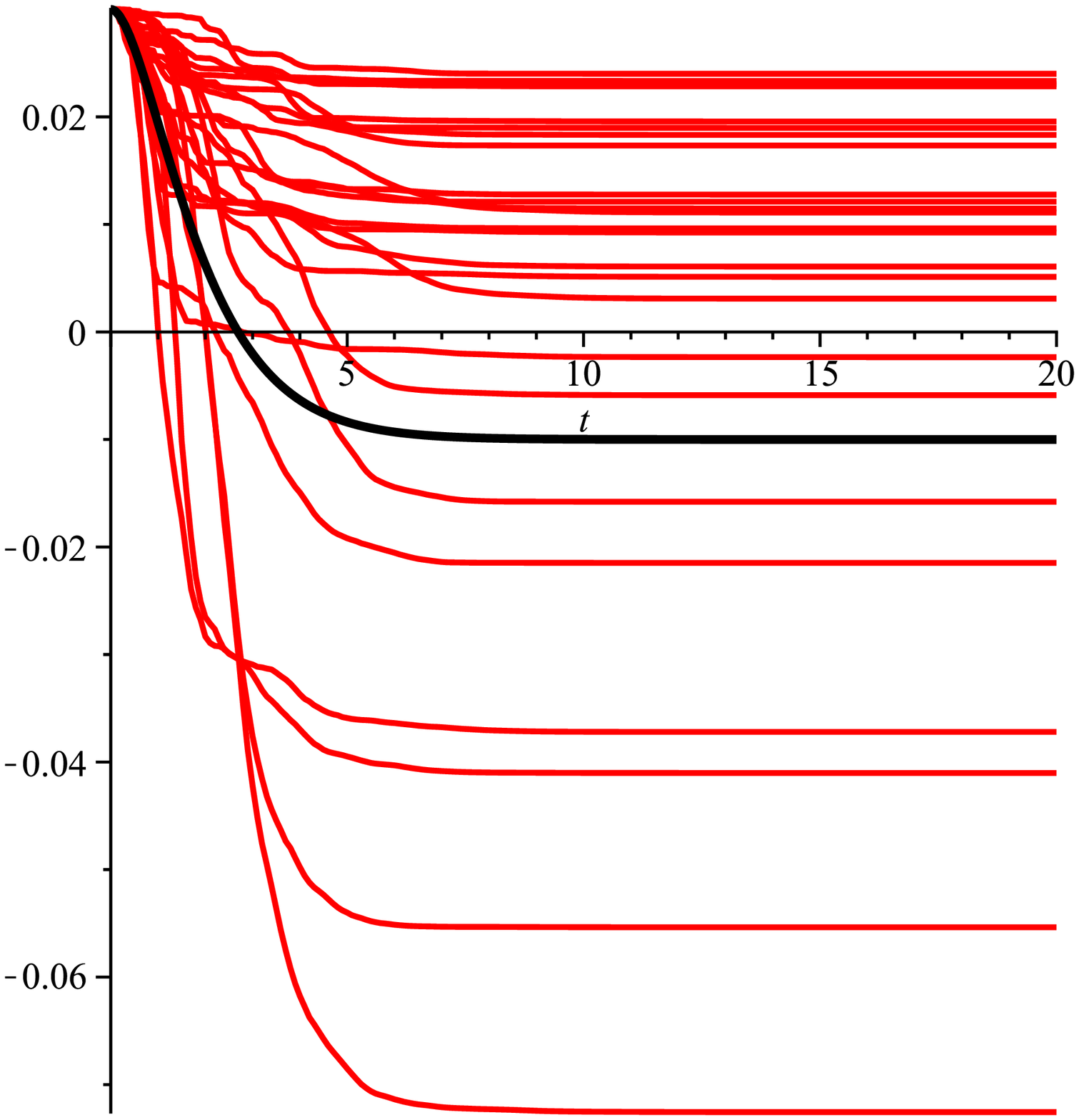}

\caption{$m=2$.}
 \end{subfigure}
\begin{subfigure}[b]{.3\linewidth}
   \includegraphics[width=0.99\textwidth]{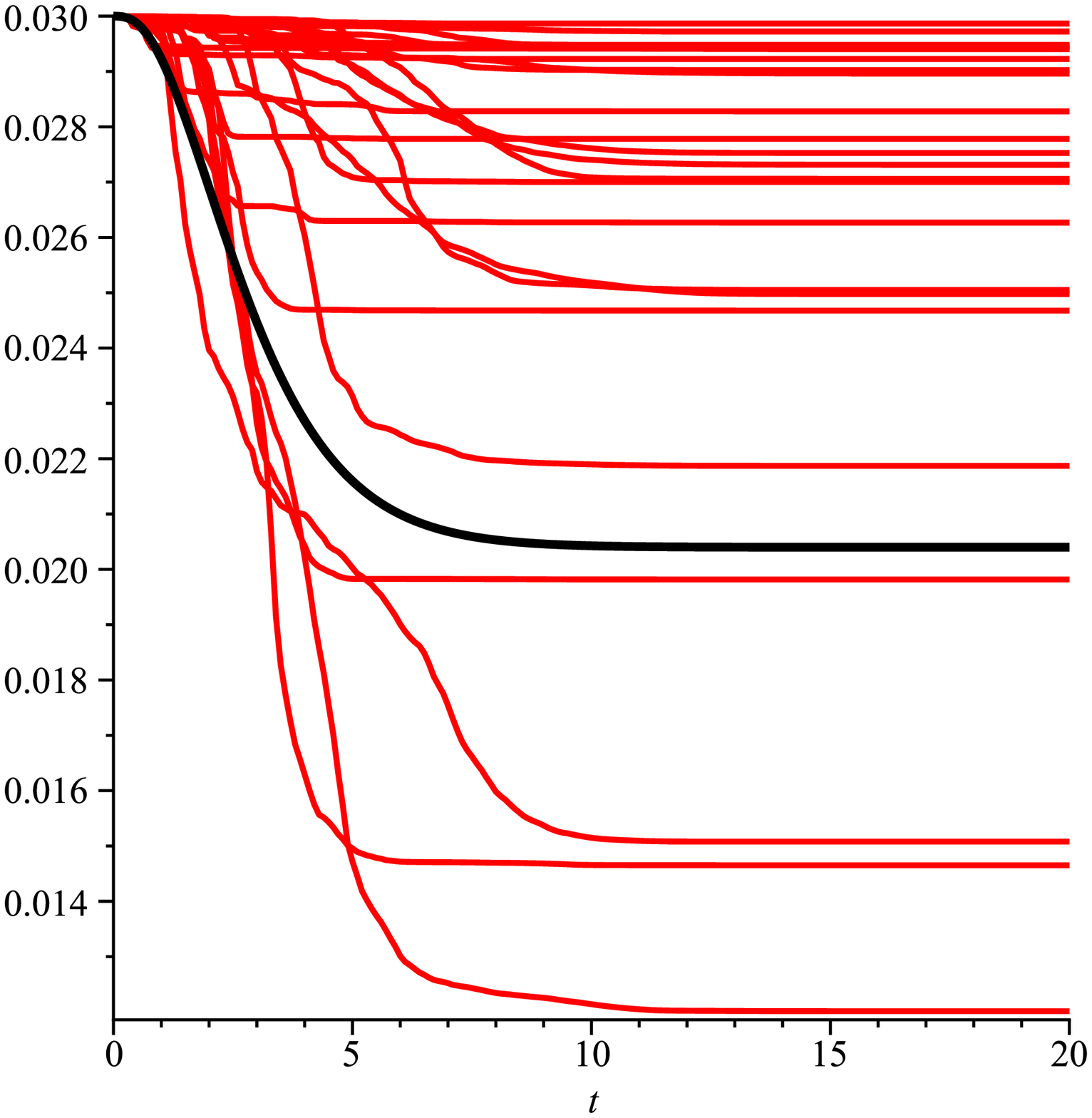}
\caption{$m=4$.}
 \end{subfigure}
 \begin{subfigure}[b]{.3\linewidth}
   \includegraphics[width=0.99\textwidth]{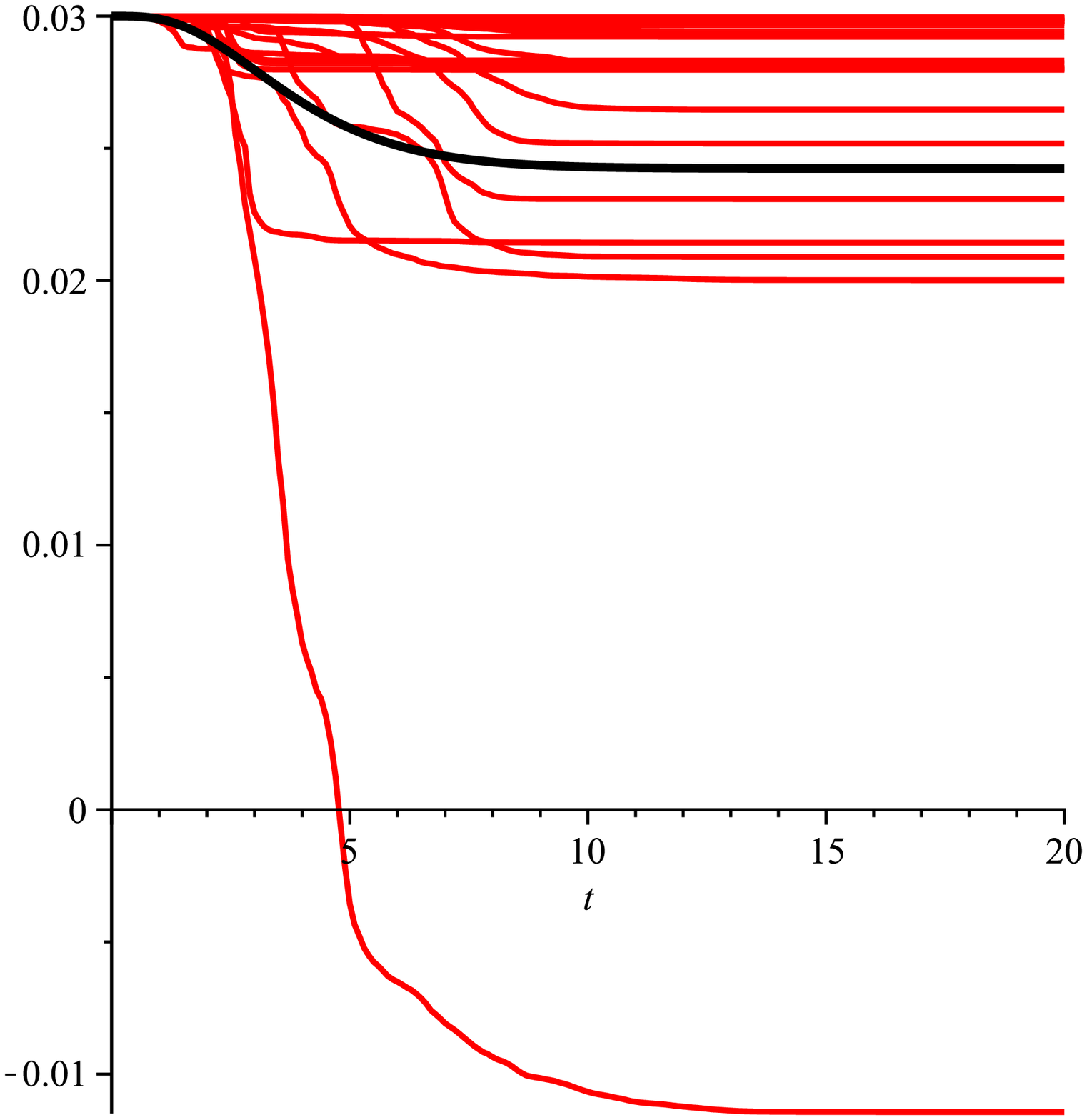}
\caption{$m=6$.}
 \end{subfigure}
 \caption{Expected value (black) and 25 realization (red) of process $R_t$ defined by \eqref{eq:kab0} with $c(t)=-\mathrm{e}^{-t}, ~k=0, ~A=0.03, ~B=0, ~a(t)=b(t)\equiv 0, ~\sigma=0.2$.}
\end{figure}

\subsection{Case study II}
 Case:  $k=n=0, ~l=2, ~\sigma \left( t \right) =\sigma, ~b \left( t \right) =-\frac{\sigma}{2}-a
 \left( t \right)  \left( W_t -\frac{t}{2} \right) ^{2}
\sigma^2.$ Then equation \eqref{eq:stochsys} reduces to
 \begin{eqnarray}\label{eq:kn0l1}
 {\mathrm d}r_t&=&c(t)\,(p_t)^m\,{\mathrm d}t \nonumber\\
{\mathrm d}p_t&=&\left[a(t)\,(p_t)^2-a(t)\left(W_t-\frac{t}{2}\right)^2\sigma^2-\frac{\sigma}{2}\right]{\mathrm d}t+\sigma\,{\mathrm d}W_t,
\end{eqnarray}
with the corresponding system
\begin{eqnarray}\label{eq:II}
 {\mathrm d}r_t&=&\sigma^m\,c(t)\,(u_t)^m\,{\mathrm d}t\nonumber\\
 {\mathrm d}z_t&=&\sigma\,a(t)\,z_t\,(2u_t-z_t)\,{\mathrm d}t.
\end{eqnarray}
In this case Theorem \ref{t:tran} is very useful.
The second equation in \eqref{eq:II} is Bernoulli differential equation with explicit general solution
$$\displaystyle z_t ={\frac {{{\rm e}^{\displaystyle\sigma\int _{0}^{t}\!a \left( u
 \right)  \left( 2\,W_u -u \right) {du}}}}{-
\sigma\,\displaystyle\int_{0}^{t}\!{{\rm e}^{\displaystyle\sigma\int _{0}^{u}\!a \left( s \right)
 \left( 2\,W_s -s \right) {ds}}}a \left( u
 \right) {du}+K}}, ~~K\in\R.
$$
Now by substituting it into the first equation, we can directly obtain a general solution
for $r_t$ by quadrature.
Denote a process $\tilde R_t$ as the solution of Cauchy problem with
$A>0, ~B\ne0$, then

\begin{align}\label{rt:fm}
\tilde R_t=A+\displaystyle\int _{0}^{t}\! \left(  \frac{ {{\rm e}^{\displaystyle\sigma\,\int _{0}^{u}\!a \left( v
 \right)  \left( 2\,W_v -v \right)\,{dv}}}}{\displaystyle -
\sigma\,\int _{0}^{u}\!{{\rm e}^{\sigma\,\displaystyle\int _{0}^{v}\!a \left( s
 \right)  \left( 2\,W_s -s \right)\,{ds}}}a \left( v
 \right) {dv}+\sigma\, \left( {\frac {c \left( 0 \right) }{B}}
 \right) ^\frac{1}{m} } +W_u -\frac{u}{2} \right)
^{m}{\sigma}^{m}c \left( u \right) {du}.
\end{align}

Exact formulation of the process $\tilde R_t$ \eqref{rt:fm} is quite difficult and thus hard to study. Nevertheless, we can use Fr\'echet derivative (functional derivative)
in order to obtain a linearized approximation. Denote as
$$F(W_t):=\frac{ {{\rm e}^{\displaystyle\sigma\,\int _{0}^{u}\!a \left( v
 \right)  \left( 2\,W_v -v \right)\,{dv}}}}{\displaystyle -
\sigma\,\int _{0}^{u}\!{{\rm e}^{\sigma\,\displaystyle\int _{0}^{v}\!a \left( s
 \right)  \left( 2\,W_s -s \right)\,{ds}}}a \left( v
 \right) {dv}+\sigma\, \left( {\frac {c \left( 0 \right) }{B}}
 \right) ^\frac{1}{m} }.$$
 Then $F(W_t)\approx F(W_0)+DF(W_0)W_t,$
 where $$F(W_0)={{\rm e}^{-\sigma\,\displaystyle\int _{0}^{u}\!a \left( v \right) v{dv}}} \left( -
\sigma\,\displaystyle\int _{0}^{u}\!{{\rm e}^{-\sigma\,\displaystyle\int _{0}^{v}\!a \left( s
 \right) s{ds}}}a \left( v \right) {dv}+\sigma\, \left( {\frac {c
 \left( 0 \right) }{B}} \right) ^{\frac{1}{m}} \right)^{-1}
$$
and
$$DF(W_0)=a \left( 0 \right) {{\rm e}^{-\sigma\,\displaystyle\int _{0}^{u}\!a \left( v
 \right) v{dv}}} \left( {\frac {c \left( 0 \right) }{B}} \right) ^{\frac{1}{m}} \left(\displaystyle \int _{0}^{u}\!{{\rm e}^{-\sigma\,\displaystyle\int _{0}^{v}\!a
 \left( s \right) s{ds}}}a \left( v \right) {dv}- \left( {\frac {c
 \left( 0 \right) }{B}} \right) ^{\frac{1}{m}} \right) ^{-2}.
$$
This approach yields the approximation of the process $\tilde R_t$:
$$\tilde R_t\approx A+\displaystyle\int _{0}^{t}\! \left( F(W_0)(u)+(DF(W_0)(u)+1)\,W_u -\frac{u}{2} \right)
^{m}{\sigma}^{m}c \left( u \right) {du}.
$$
Using similar techniques as for study case I, one can obtain an approximation of expectation and variance of given process.
From computational point of view result \eqref{rt:fm} can be hardly used, e.g.~in {\rm The Finance package} in Maple,
one can not plot paths of doubly integrated Wiener process. Thus one should use approximation of such process involving linearization of $F(w_t)$.
Notice here that for specific values of this model has increase of parameter $m$
opposite effect as for the model \eqref{eq:kab0}, see Figure 8 (a).

\begin{figure}[h!]\label{fig:b}
\centering
\begin{subfigure}[b]{.3\linewidth}
   \includegraphics[width=0.99\textwidth]{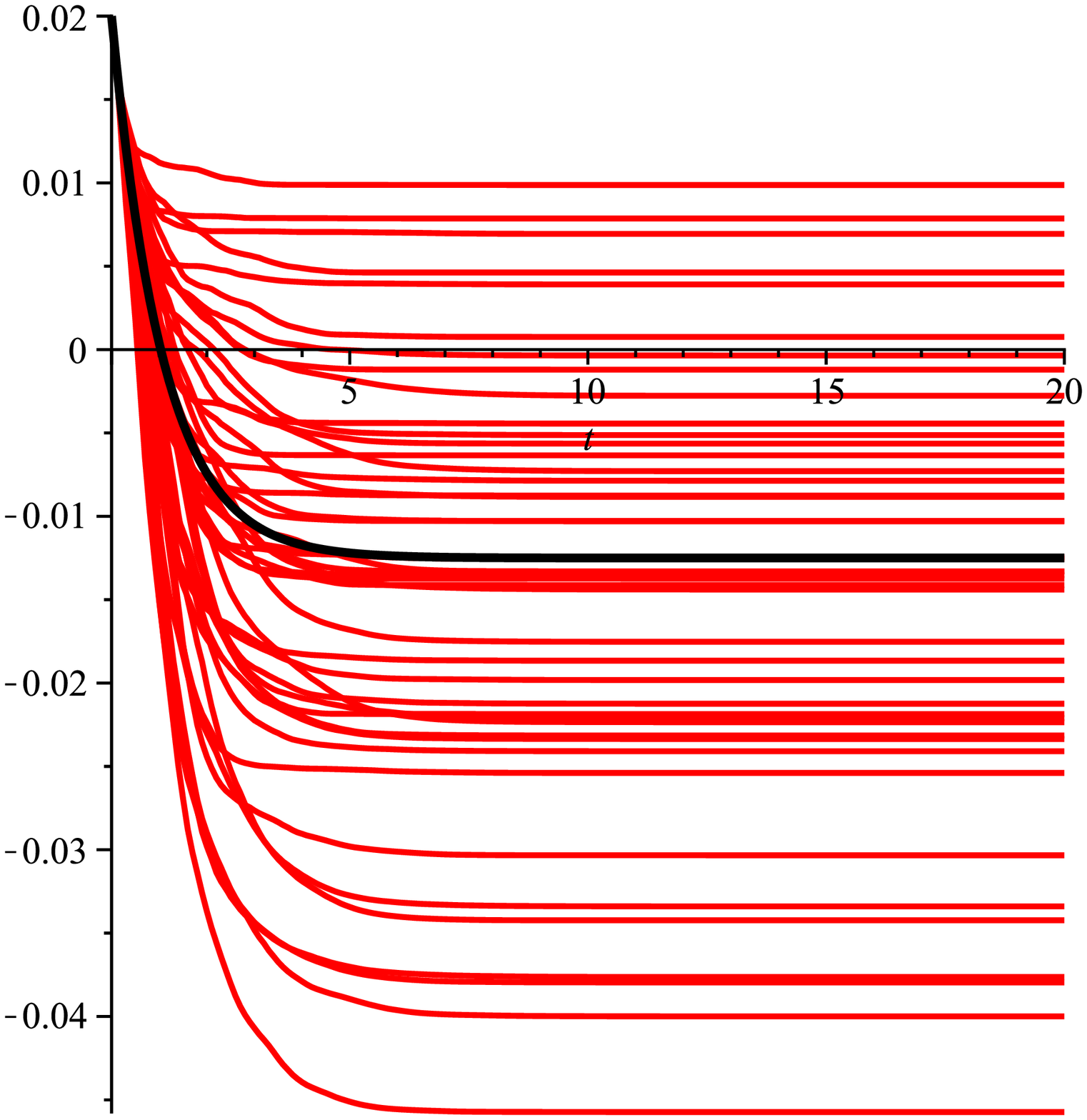}

\caption{$B=0.03$.}
 \end{subfigure}
\begin{subfigure}[b]{.3\linewidth}
   \includegraphics[width=0.99\textwidth]{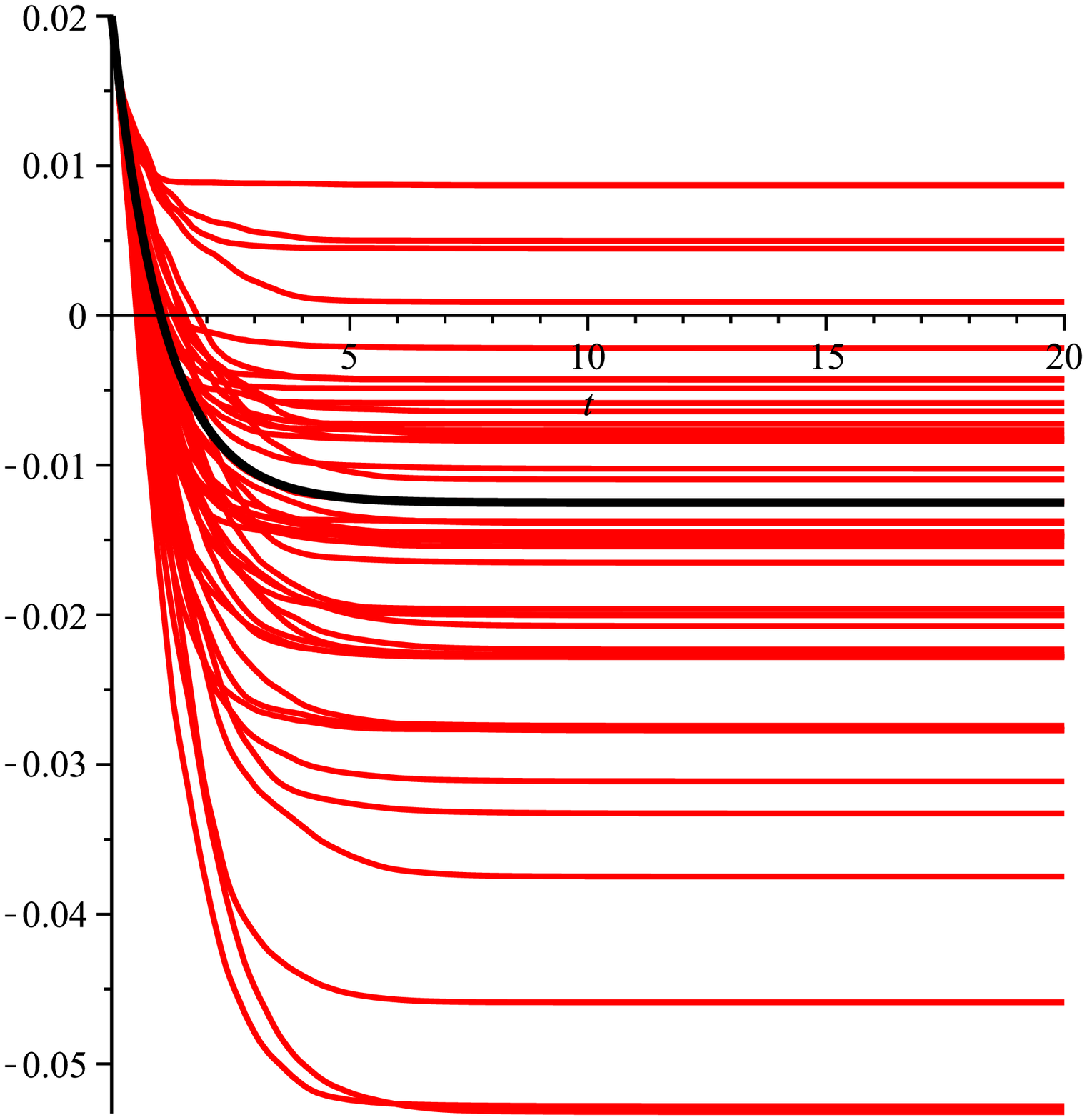}
\caption{$B=-0.03$.}
 \end{subfigure}
 \begin{subfigure}[b]{.3\linewidth}
   \includegraphics[width=0.99\textwidth]{caseexpt}
\caption{$B=0$.}
 \end{subfigure}
 \caption{Expected value (black) and 25 realization (red) of process
 \eqref{eq:kab0} with $c(t)=-\mathrm{e}^{-t}, ~k=0, ~A=0.02, ~m=2, ~a(t)=b(t)\equiv 0, ~\sigma=0.05$.}
\end{figure}

\begin{figure}[h!]\label{fig:m}
\centering
\begin{subfigure}[b]{.45\linewidth}
   \includegraphics[width=0.99\textwidth]{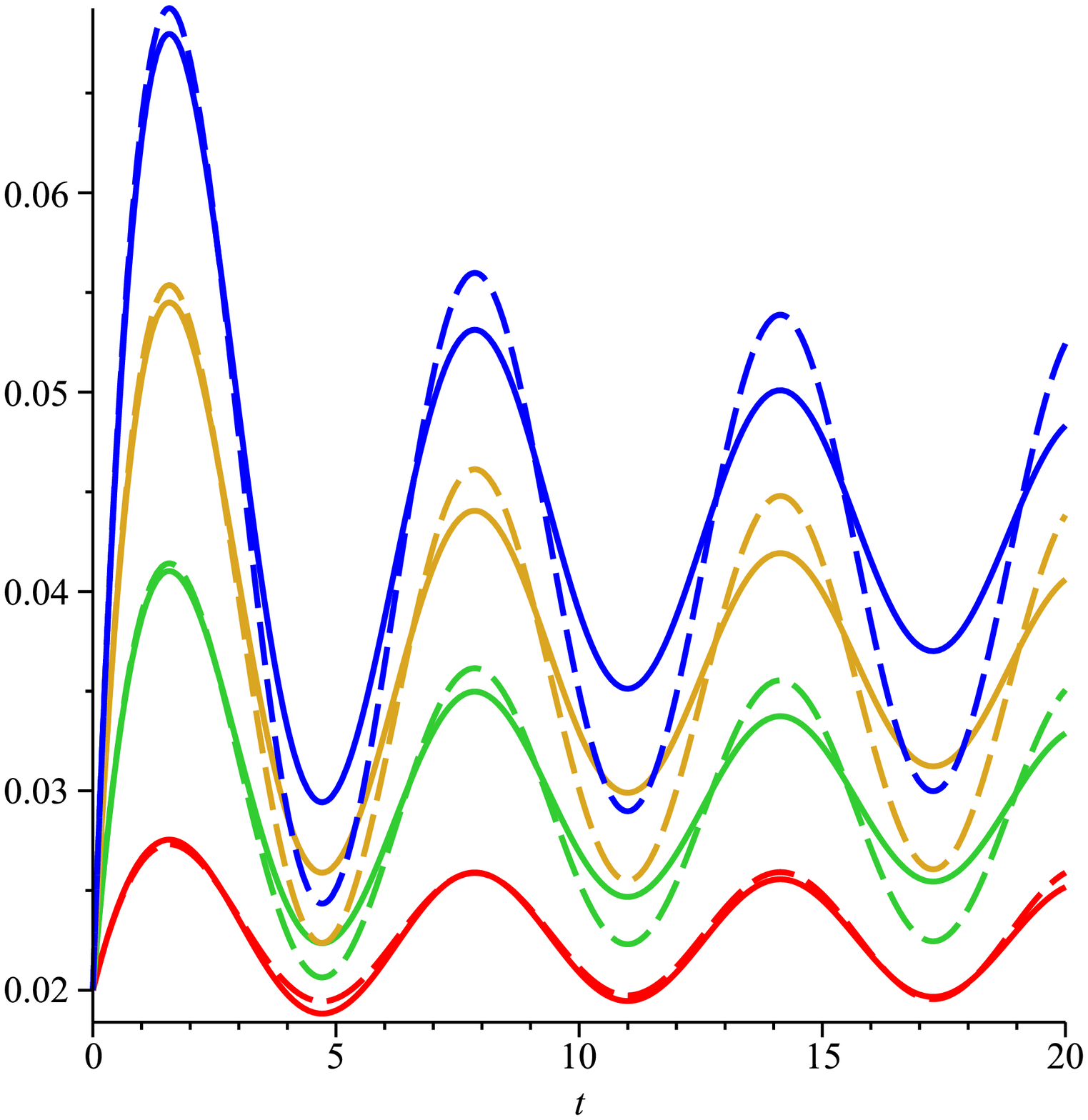}
\caption{$B\in\{0.01, 0.03, 0.05, 0.07\}, ~~m\in\{2,6\},$\\$c(t)=\frac{\cos{t}}{t+1}$.}
 \end{subfigure}
\begin{subfigure}[b]{.45\linewidth}
   \includegraphics[width=0.99\textwidth]{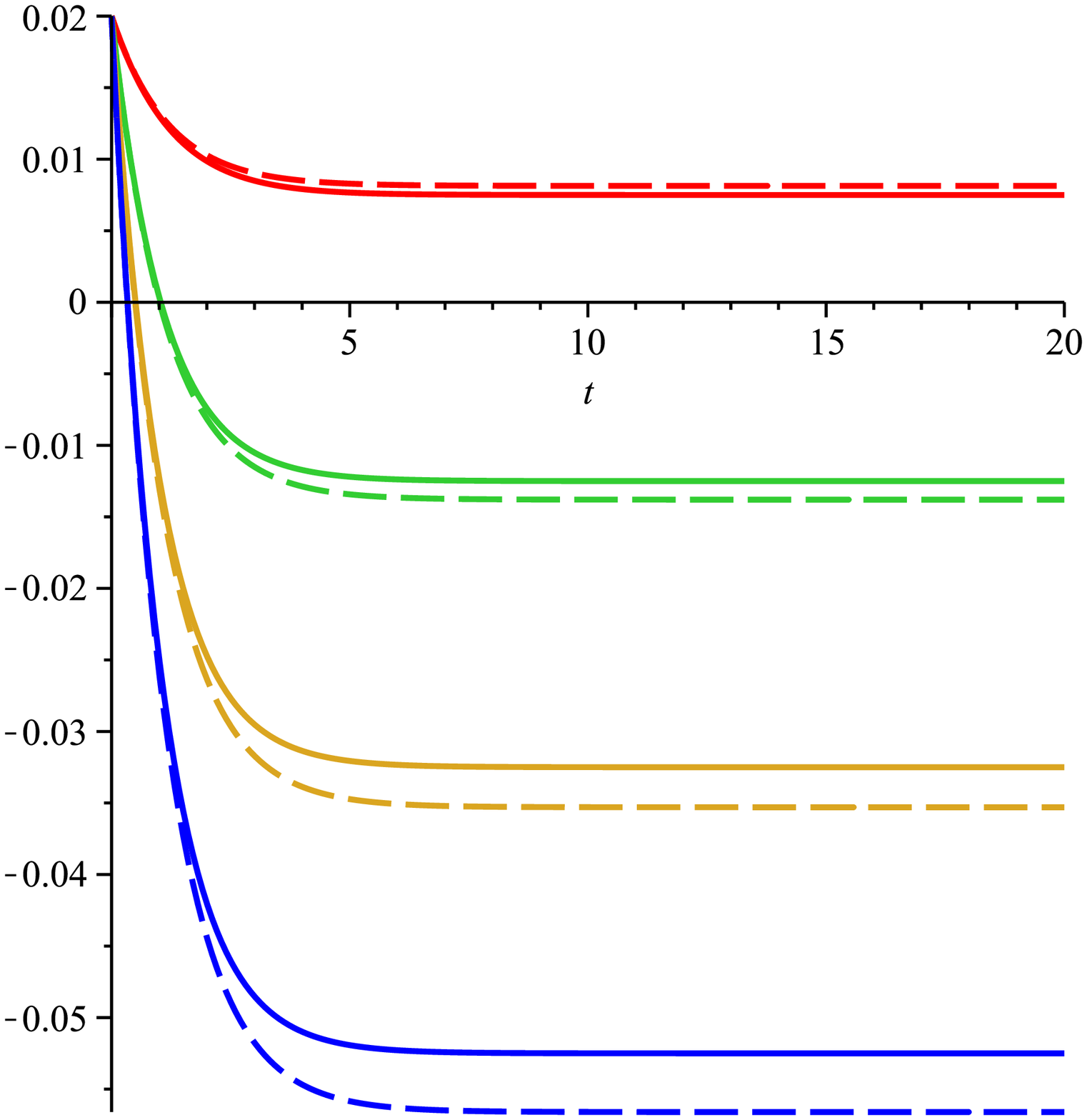}
\caption{$B\in\{0.01, 0.03, 0.05, 0.07\}, ~~m\in\{2,6\},$ \\ $c(t)=-\mathrm{e}^{-t}$.}
 \end{subfigure}
 \caption{Influence of parameters $B, ~m$ ~on the trajectory of the expected value (see Lemma \ref{lm:e}) of process $R_t$ defined by \eqref{eq:kab0}
 with $k=0, ~A=0.02, ~a(t)=b(t)\equiv 0, ~\sigma=0.05$ ($m=6$ - dashed line).}
\end{figure}

\begin{Ex}
 For $m=2$ we have
 $$E[\tilde R_t]\approx A+\int_0^t\left[\left(F(W_0)(u)-\frac{u}{2}\right)^2+(DF(W_0)(u)+1)^2\,u\right] {\sigma}^{m}c \left( u \right) {du}$$
 and
 $$Var[\tilde R_t]\approx \sigma^2\int_0^t\int_0^tc(u)\,c(s)\,\left[
 \left(F(W_0)(u)-\frac{u}{2}\right)^2\left(F(W_0)(s)-\frac{s}{2}\right)^2
+\right.$$
$$+\left(F(W_0)(u)-\frac{u}{2}\right)^2(DF(W_0)(u)+1)^2\,s
 +\left(F(W_0)(s)-\frac{s}{2}\right)^2(DF(W_0)(s)+1)^2\,u+$$
 $$+
 4\left(F(W_0)(u)-\frac{u}{2}\right)\left(F(W_0)(s)-\frac{s}{2}\right)(DF(W_0)(s)+1)(DF(W_0)(u)+1)\,\mathcal{E}_1(u,s)+
 $$
 $$+\left.\vphantom{\left(F(W_0)(s)-\frac{s}{2}\right)^2}(DF(W_0)(s)+1)^2(DF(W_0)(u)+1)^2\,\mathcal{E}_2(u,s)+\right]\,\mathrm{d}u\mathrm{d}s-$$
 $$\left(\int_0^t\left[\left(F(W_0)(u)-\frac{u}{2}\right)^2+(DF(W_0)(u)+1)^2\,u\right] {\sigma}^{m}c \left( u \right) {du}\right)^2.$$

\end{Ex}

\begin{figure}[h!]
\centering
\begin{subfigure}[b]{.45\linewidth}
   \includegraphics[width=0.99\textwidth]{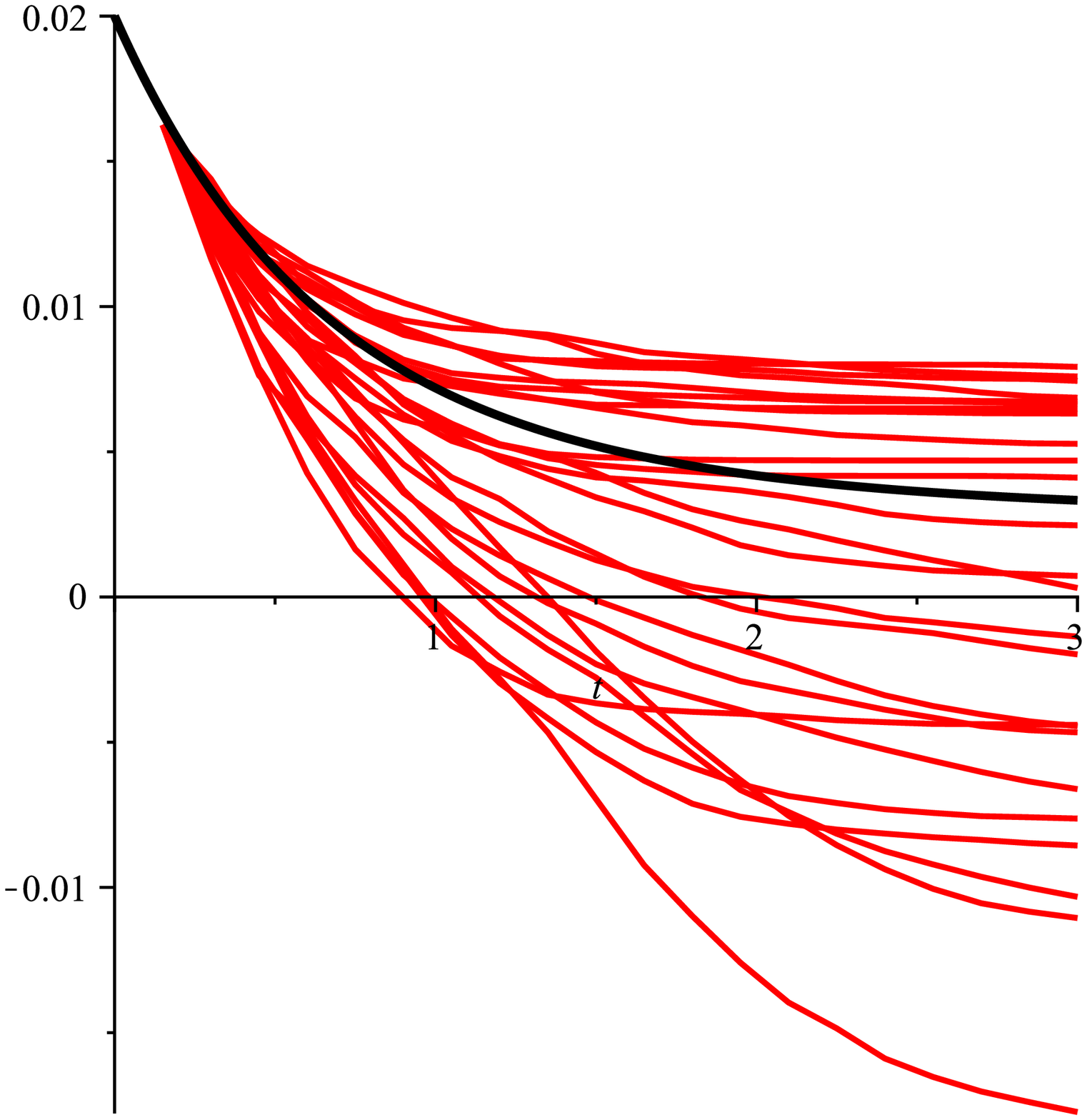}
\subcaption{$\sigma=0.05$.}
 \end{subfigure}
\begin{subfigure}[b]{.45\linewidth}
   \includegraphics[width=0.99\textwidth]{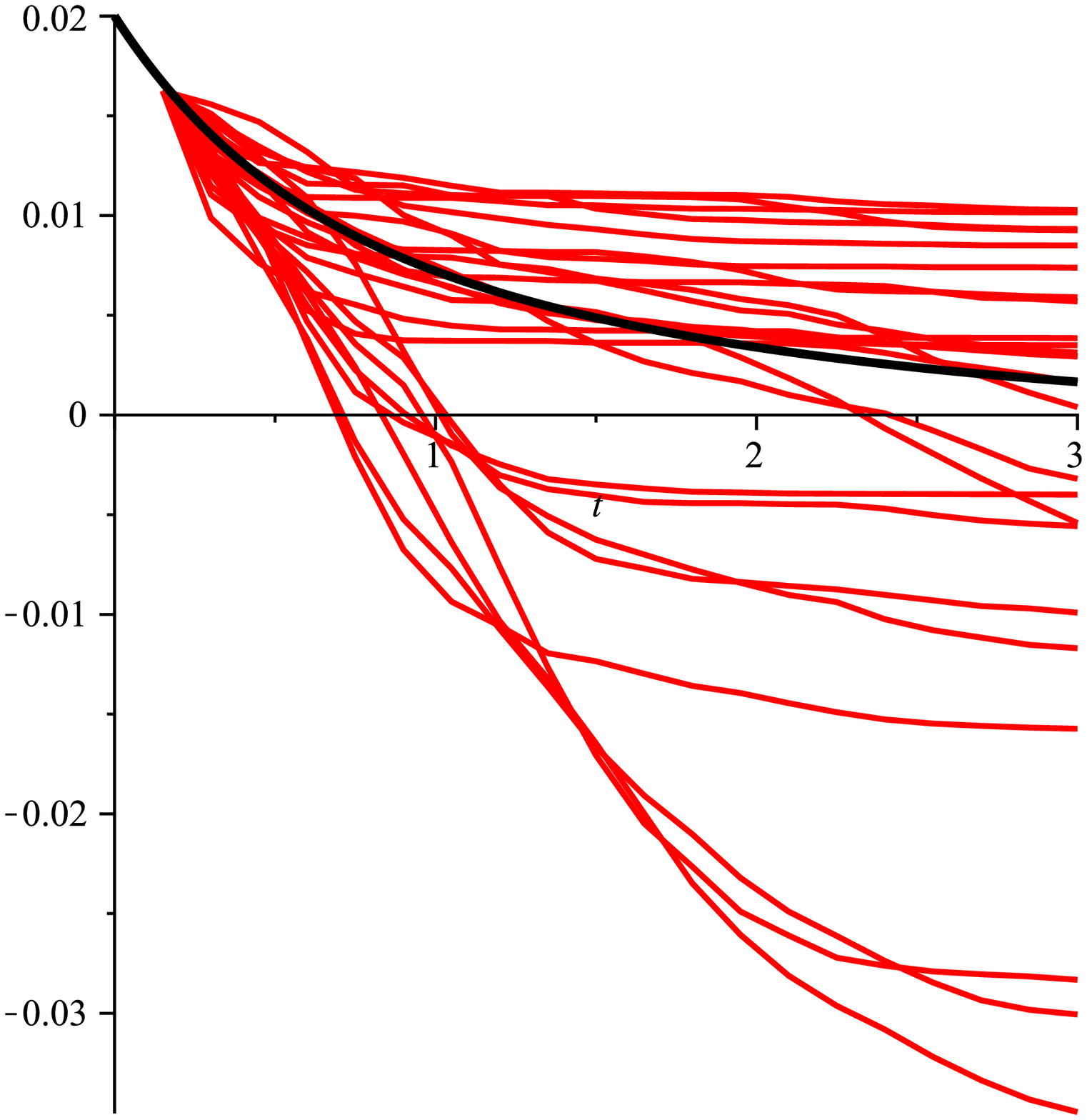}
\subcaption{$\sigma=0.1$.}
 \end{subfigure}
 \caption{Expected value (black) and 25 realization (red) of
 process \eqref{eq:kn0l1} (with linearized approximation $F(W_t)$)
  with $c(t)=-\mathrm{e}^{-t}, ~k=n=0, ~l=2, ~A=0.02, ~B=-0.025, ~m=2, ~a(t)\equiv -1$.}
\end{figure}

\begin{figure}[h!]
\centering
\begin{subfigure}[b]{.45\linewidth}
   \includegraphics[width=0.99\textwidth]{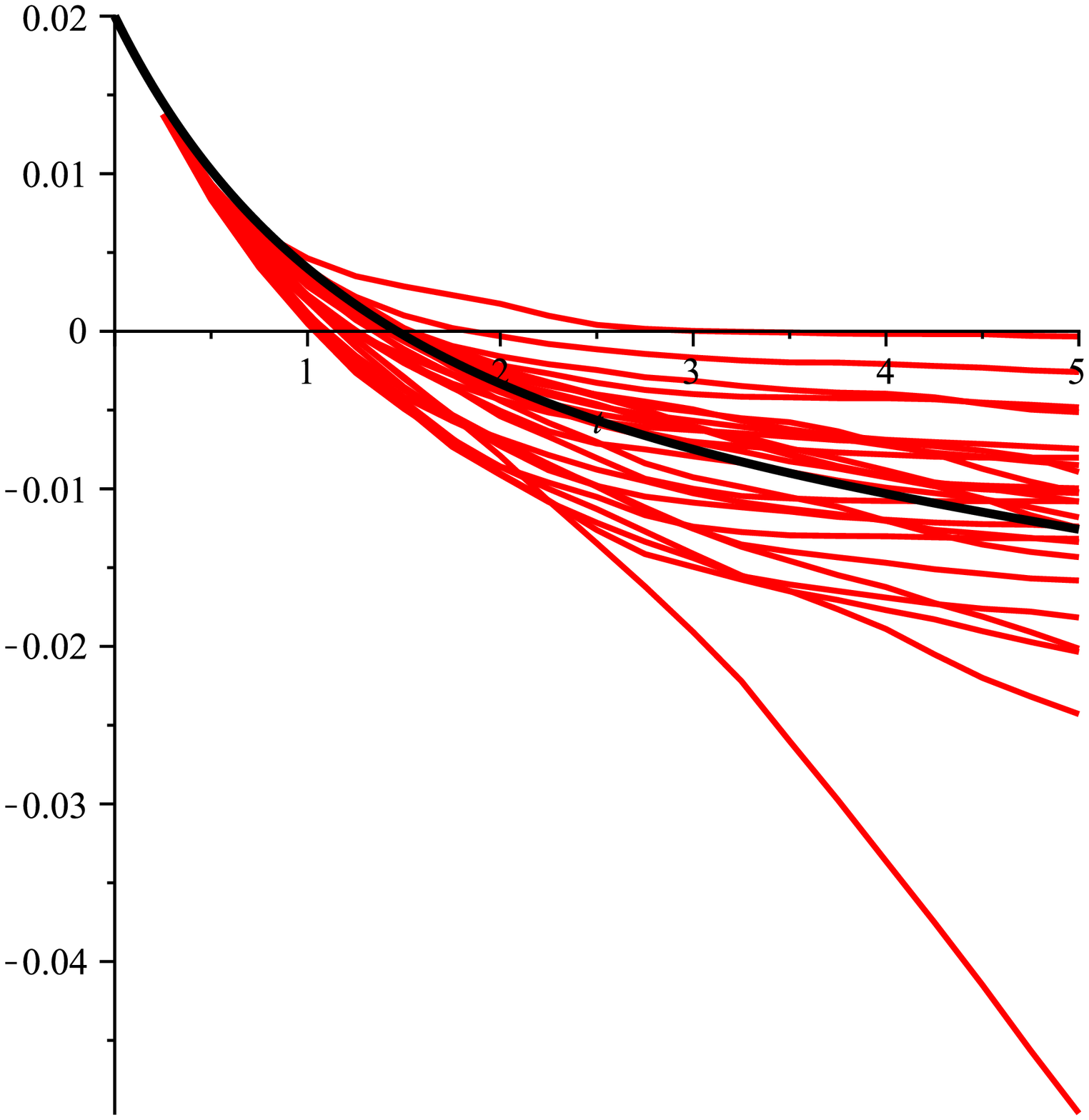}
\subcaption{$\sigma=0.05$.}
 \end{subfigure}
\begin{subfigure}[b]{.45\linewidth}
   \includegraphics[width=0.99\textwidth]{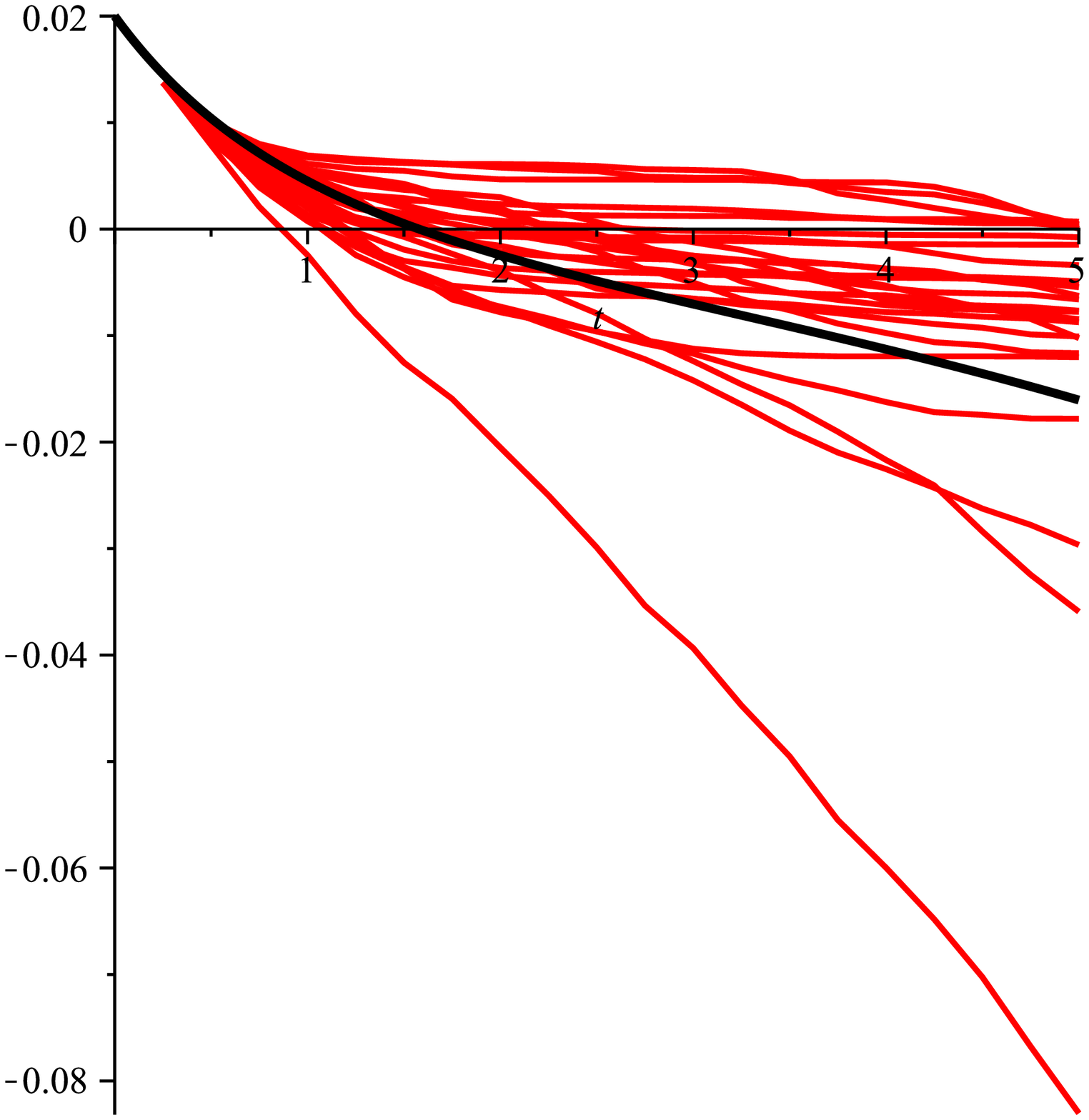}
\subcaption{$\sigma=0.1$.}
 \end{subfigure}
 \caption{Expected value (black) and 25 realization (red) of
 process \eqref{eq:kn0l1} (with linearized approximation $F(W_t)$)
  with $c(t)\equiv -0.1, ~k=n=0, ~l=2, ~A=0.02, ~B=-0.025, ~m=2, ~a(t)\equiv -1$.}
\end{figure}

\begin{figure}[h!]\label{fig:FWm}
\centering
\begin{subfigure}[b]{.3\linewidth}
   \includegraphics[width=0.99\textwidth]{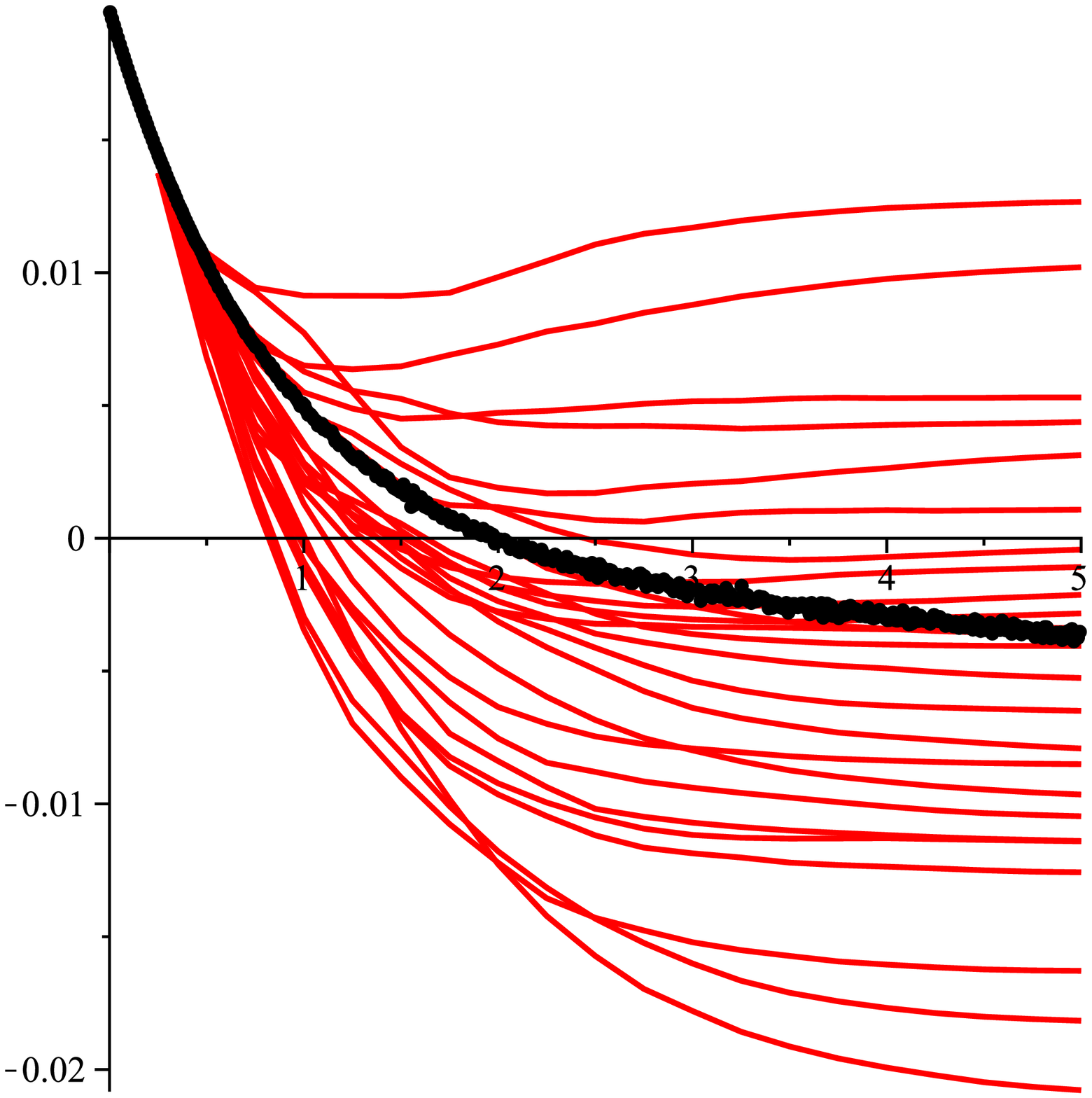}

\caption{$m=1$.}
 \end{subfigure}
\begin{subfigure}[b]{.3\linewidth}
   \includegraphics[width=0.99\textwidth]{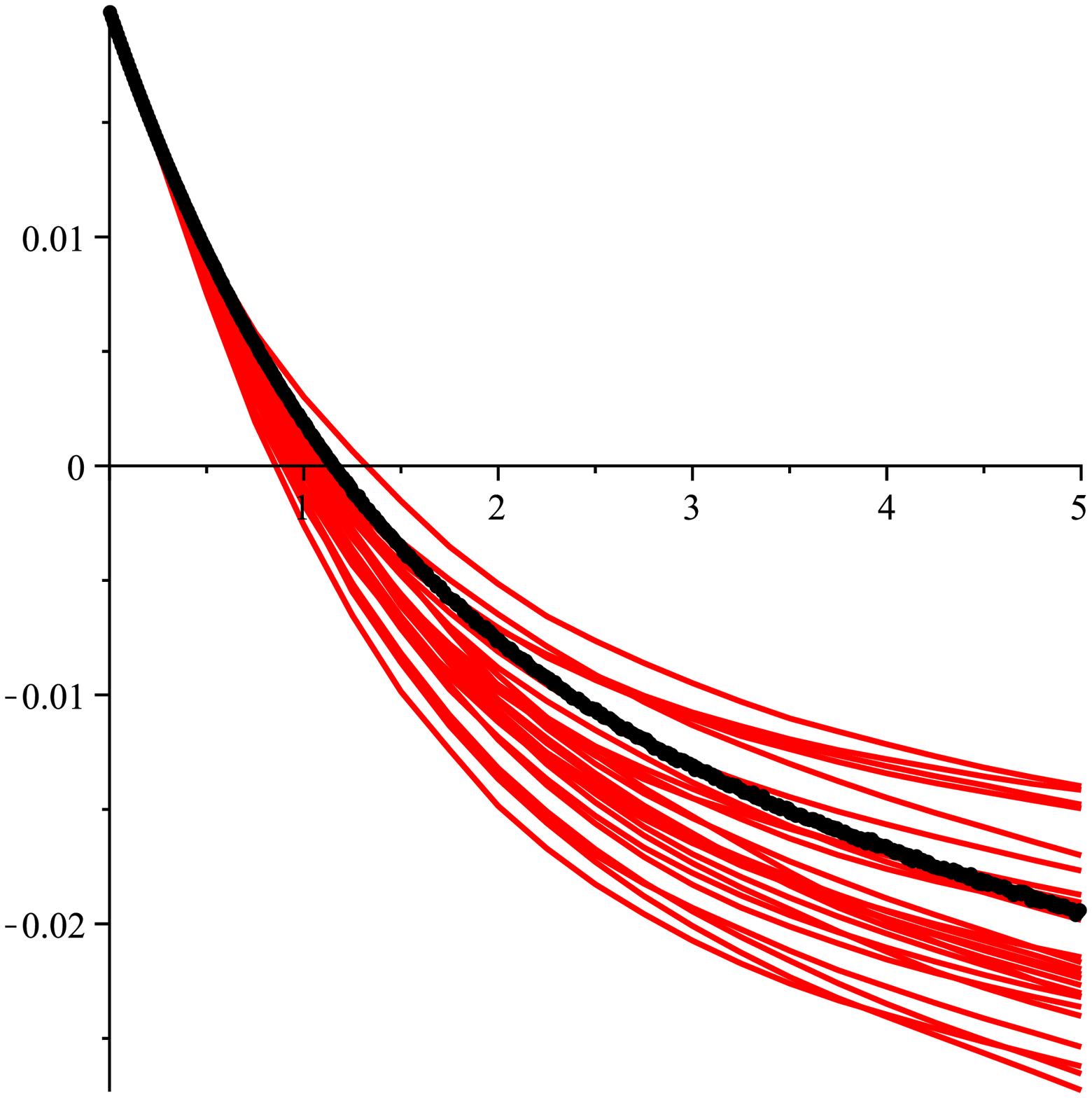}
\caption{$m=2$.}
 \end{subfigure}
 \begin{subfigure}[b]{.3\linewidth}
   \includegraphics[width=0.99\textwidth]{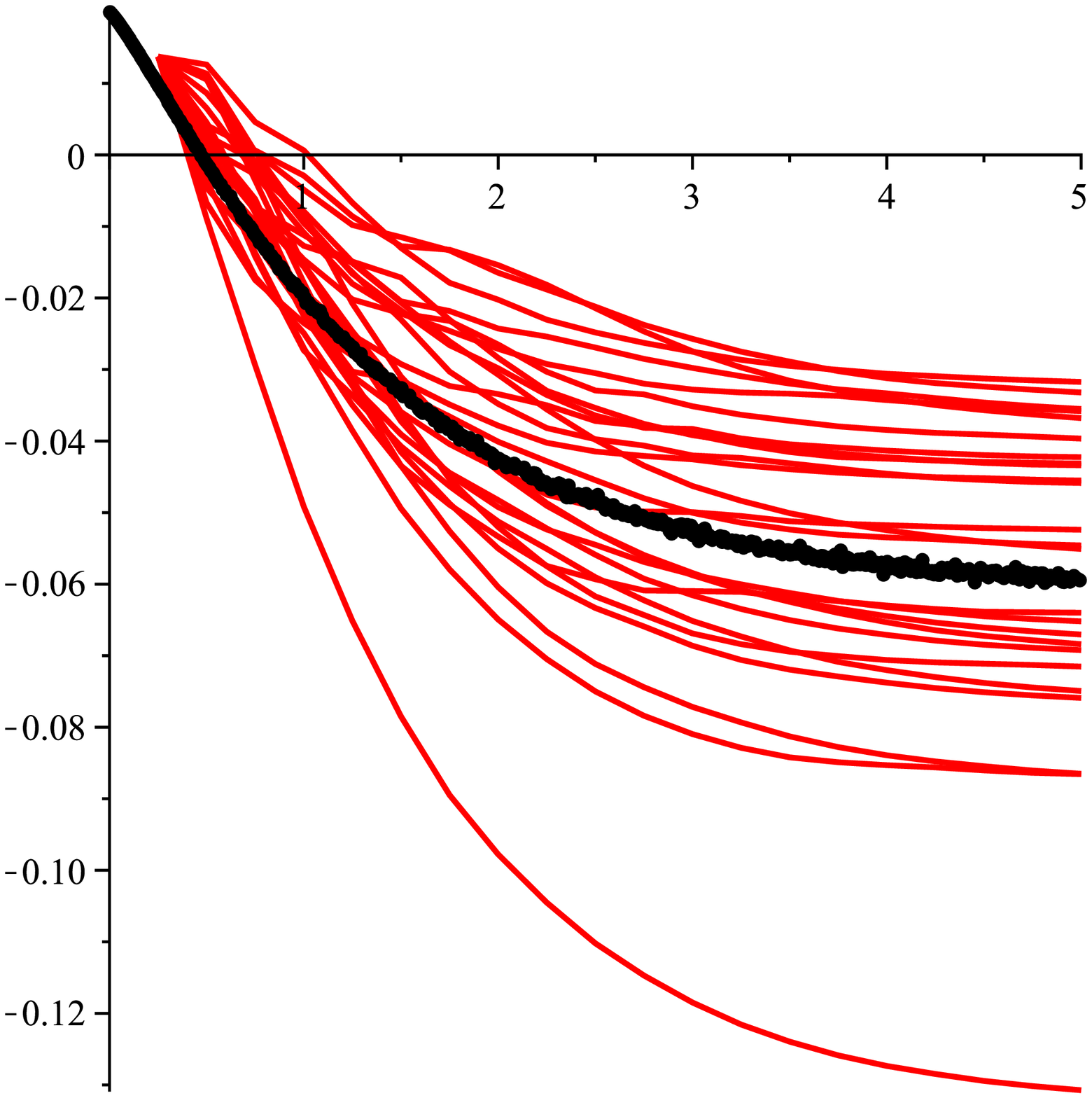}
\caption{$m=\frac{1}{2}$.}
 \end{subfigure}
 \caption{Expected value (black) and 25 realization (red) of
 process \eqref{eq:kn0l1} (with linearized approximation $F(W_t)$)
  with
  $c(t)=-\mathrm{e}^{-t}, ~k=0, ~A=0.02, ~B=-0.025, ~a(t)=b(t)\equiv 1, ~\sigma=0.01$.}
\end{figure}

\section{Remarks and interpretations}

\begin{rem}
In Case study I, $B=0$, we consider constant volatility $\sigma(t)=\sigma>0$ and
an interesting fact is that for
\emph{dumping scale function} $c(u)>0$ ~(e.g.~$c(u)=\exp(-u),u>0$)  we have
an estimation of $Var[R_t]$ from below by Gamma kernel:
$$Var[R_t]\geq \sigma^{2m}\left[F(t)\displaystyle\frac{(m!)^2}{2^{m}}\sum_{j=0}^\frac{m}{2}\frac{2^{2j}}{(2j)!\left[\left(\frac{m}{2}-j\right)!\right]^2}-\frac{m!}{2^\frac{m}{2}\left(\frac{m}{2}\right)!}\int_0^tc(u)\,u^\frac{m}{2}\mathrm{d}u\right]\geq$$
$${\sigma}^{2m}  m!\, \left( {\frac {  m!\,{2}^{m}\,\Gamma  \left( \frac{1}{2}+m \right) F \left( t
 \right) }{\sqrt {\pi }\,\Gamma  \left( m+1 \right)
^{2}}}-{\frac {{\displaystyle{t}^{\frac{m}{2}+1}\,\max_{u\in[0,t]}\,c(u)}}{ \left( \frac{m}{2}+1 \right)!\,{2}^{\frac{m}{2}}}} \right)
$$
for $m$ even, $F(t)=t^2$. Thus for a finite time $t>0$ the general form of limit is $\infty$. Similar argumentation
can be used for the odd case. We have infinite \emph{dumping power} $m$ for fixed time,
that with increasing power $m$
expectation of the modelled interest rate force can exceed negative values. On the other hand
small shift of volatility parameter $\sigma$ could have a reverse effect.
This phenomenon is caused, because we also get an infinite expected value of $E[R_t]$ for increasing $m$.
In contrast to that, for $m$ close to zero we have $E[R_t]$ close to $R_0+1=A+1$ for all $\sigma>0$ and $c(u)=\exp(-u),u>0.$\\
For example, for $B\ne 0$ and $c(t)\equiv c\not\equiv 0$ it can be shown that
$$E[R_t]=A+B\,t{c}^{1-m}\,
{{}_3F_1\left(\left[1,-\frac{m}{2},\frac{1}{2}-\frac{m}{2}\right];\,[2];\,\left[2\,{c}^{2}{\sigma}^{2}t{B}^{-\frac{2}{m}}\right]\right)},
$$
where ${}_3F_1$ is a generalized hypergeometric function.
\end{rem}

\paragraph*{Acknowledgement.}
First author was partially supported by grant VEGA M\v S SR 1/0344/14.
Second author gratefully acknowledges support of ANR project Desire FWF I 833-N18.
Corresponding author acknowledges Fondecyt Proyecto Regular No 1151441. Last but not the least, the authors are very grateful to the
editor and the reviewers for their valuable comments.

\clearpage

\bibliographystyle{apalike} 
 \bibliography{LMJ_Refs}

\vspace{1cm}
\noindent Milan Stehl\'\i k (Corresponding author) \\ Institute of Statistics  \\Universidad de Valpara\'iso, Valpara\'iso, Chile. \\
Department of Applied Statistics\\
Johannes Kepler University, Linz, Austria.\\
Email:  milan.stehlik@jku.at\\
Tel.+56 32 2654680, +4373224686808; Fax: +4373224686800\\

\noindent Philipp Hermann \\
Department of Applied Statistics\\
Johannes Kepler University, Linz, Austria.\\
Email: philipp.hermann@jku.at\\

\noindent Jozef Kise{\v l}\'ak\\
Institute of Mathematics,\\
P.J. \v Saf\'arik Univesity in Ko\v sice, Slovak Republic\\
Email: jozef.kiselak@upjs.sk

\appendix
\section{Appendices}
\subsection{First appendix}
We write $|Z|^2=\sum_{i,j}|Z_{ij}|^2$ for matrix $Z$.
\begin{thm}[\cite{Oksendal}, Theorem. 5.2.1.]\label{exi}
 Let $T>0$ and $b(\cdot,\cdot): [0,T]\times {\mathbb R}^n\to {\mathbb R}^n$,
$\Sigma(\cdot,\cdot): [0,T]\times {\mathbb R}^n\to {\mathbb R}^{n\times m}$ be
measurable functions satisfying
\begin{equation*}
 |b(t,x)|+|\Sigma(t,x)|\leq C(1+|x|);
\end{equation*}
$(x,t)\in{\mathbb R}^n\times[0,T]$ for some constant $C$, and such that
\begin{equation*}
 |b(t,x)-b(t,y)|+|\Sigma(t,x)-\Sigma(t,y)|\leq D|x-y|;
\end{equation*}
$x,y\in{\mathbb R}^n, ~~t\in[0,T]$ for some constant $D$. Let $Z$ be a random variable which is independent of the
$\sigma$-algebra ${\mathcal{F}}^{(m)}_\infty$ generated  by $W_s(\cdot), ~s\geq 0$ and such that
$$E\left[|Z|^2\right]<\infty.$$
Then the stochastic differential equation
\begin{eqnarray}
 {\mathrm d}X_t&=&b(t,X_t)\,{\mathrm d}t+\Sigma(t,X_t)\,{\mathrm d}W_t, \quad 0\leq t\leq T, \nonumber\\
X_0&=&Z \nonumber
\end{eqnarray}
has a unique $t$-continuous solution $X_t(\omega)$ with the property that $X_t(\omega)$
is adapted to the filtration ${\mathcal{F}}^Z_t$ generated by $Z$ and $W_s(\cdot); ~s\leq t$
and $$E\left[\int_0^T|X_t|^2{\mathrm d}t<\infty\right].$$
\end{thm}

\begin{thm}[It\^os lemma]\label{ito}
Let $X_t$ be a process  given by
\begin{equation}
 \mathrm\mathrm{d}X_t=f(X_t,t)\mathrm\mathrm{d}t+\sigma(X_t,t)\mathrm\mathrm{d}W_t.
\end{equation}
Let $\psi(X_t,t)\in C(\mathbb{R}^n\times[0,\infty))$, then for a transformation $Z_t=[\psi_1(X_t,t),\dots,\psi_n(X_t,t)]$, $Z_t$
is again an It\^o process given by
\begin{equation}
 \mathrm\mathrm{d}Z_t=\frac{\partial \psi}{\partial t}(X_t,t)\mathrm\mathrm{d}t+
 \sum_{i=1}^n\frac{\partial \psi}{\partial x_i}(X_t,t)\mathrm\mathrm{d}X_{i,t}+
 \frac{1}{2}\sum_{i=1}^n\sum_{j=1}^n\frac{\partial^2 \psi}{\partial x_i\partial x_j}(X_t,t)\mathrm\mathrm{d}X_{j,t}\mathrm\mathrm{d}X_{i,t},
\end{equation}
where $\mathrm\mathrm{d}X_{j,t}\mathrm\mathrm{d}X_{i,t}$ is calculated according  to the standard rules
$(\mathrm\mathrm{d}t)^2=\mathrm\mathrm{d}t\,\mathrm\mathrm{d}W_{i,t}=0, ~\mathrm\mathrm{d}W_{j,t}\,\mathrm\mathrm{d}W_{i,t}=0, ~j\ne i, ~\mathrm\mathrm{d}W_{i,t}\,\mathrm\mathrm{d}W_{i,t}=\mathrm\mathrm{d}t$.
\end{thm}

\begin{lm}[Higher moments of Wiener process, see e.g.~\cite{Oksendal}] \label{lemmastar}
For standard Brownian motion (Wiener process). Following formula for higher moments holds:
\begin{align*} \beta_k(t) &:= E[W_t^k], \quad k=0,\dots, t\leq0\\
\beta_k(t) &= \frac{k(k-1)}2 \int_0^t \beta_{k-2}(s) \dd s, \quad k\leq2, \quad \text{which implies} \\
E[W_t^k] &= \begin{cases} 0 & \mbox{k is odd,} \\ \frac{k!}{2^{\frac k2}(\frac k2)!} t^{\frac k2}, & \mbox{k is even}.\end{cases}
\end{align*}
\end{lm}

\begin{lm}\label{fan}
Let $z$ be bivariate normal with zero mean and $s_1, s_2$ be nonnegative integers,then
$$ E[Z_1^{s_1}\,Z_2^{s_2}]=\begin{cases}
   0, & \mbox{if} ~~s_1+s_2 ~\mbox{is odd},\\
   \displaystyle\sigma_1^{s_1}\,\sigma_2^{s_2}\frac{s_1!\,s_2!}{2^{\frac{s_1+s_2}{2}}}\sum_{j=0}^\frac{\min{(s_1,s_2)}}{2}\frac{(2\mathrm{cor}(Z_1,Z_2))^{2j}}{(2j)!\left(\frac{s_1}{2}-j\right)!\left(\frac{s_2}{2}-j\right)!}, & \mbox{if} ~~s_1, s_2 ~\mbox{are even},\\
      \displaystyle\sigma_1^{s_1}\,\sigma_2^{s_2}\frac{s_1!\,s_2!}{2^{\frac{s_1+s_2-2}{2}}}\sum_{j=0}^\frac{\min{(s_1-1,s_2-1)}}{2}\frac{(2\mathrm{cor}(Z_1,Z_2))^{2j+1}}{(2j+1)!\left(\frac{s_1-1}{2}-j\right)!\left(\frac{s_2-1}{2}-j\right)!}, & \mbox{if} ~~s_1, s_2 ~\mbox{are odd}.
  \end{cases}
$$
\end{lm}

\end{document}